\documentclass[runningheads]{llncs}

\usepackage[T1]{fontenc}
\usepackage{amsmath}

\usepackage{amssymb}
\usepackage{mathtools}
\usepackage{xspace}
\usepackage{xcolor}
\usepackage{booktabs}
\usepackage[most]{tcolorbox}
\usepackage{enumitem}
\usepackage{thmtools}
\setlist{nosep, leftmargin=*}
\setlist{itemsep=1pt, topsep=3pt, leftmargin=*}

\AtBeginDocument{%
  \setlength{\abovedisplayskip}{2pt plus 1pt minus 1pt}%
  \setlength{\belowdisplayskip}{2pt plus 1pt minus 1pt}%
  \setlength{\abovedisplayshortskip}{1pt plus 1pt}%
  \setlength{\belowdisplayshortskip}{1pt plus 1pt}%
}

\makeatletter
\renewcommand\section{\@startsection{section}{1}{\z@}%
  {-6\p@ \@plus -2\p@ \@minus -2\p@}%
  {3\p@ \@plus 1\p@}%
  {\normalfont\bfseries\large}}
\renewcommand\subsection{\@startsection{subsection}{2}{\z@}%
  {-4\p@ \@plus -1\p@}%
  {2\p@ \@plus 1\p@}%
  {\normalfont\bfseries}}
\renewcommand\subsubsection{\@startsection{subsubsection}{3}{\z@}%
  {-3\p@ \@plus -1\p@}%
  {1\p@ \@plus 1\p@}%
  {\normalfont\bfseries}}
\makeatother

\newcommand{\mypara}[1]{\smallskip\noindent\textbf{#1.}}
\newcommand{\mysubpara}[1]{\emph{#1.}}

\newcommand{\remove}[1]{}

\begin{document}
\title{Constitutional Governance in Metric Spaces}
\author{Ehud Shapiro\inst{1} \and Nimrod Talmon\inst{2}}
\authorrunning{Shapiro and Talmon}
\institute{London School of Economics and Weizmann Institute of Science \and Ben-Gurion University and Input Output}
\maketitle

\begin{abstract}
Computational social choice and algorithmic decision theory offer rich aggregation theory but no end-to-end process for egalitarian self-governance: aggregation, deliberation, amendment, and consensus are each considered in isolation, with key metric-space aggregators being NP-hard. 
Here, we propose \emph{constitutional governance in metric spaces}, integrating these stages into a protocol for constitutional governance. A community's \emph{legal corpus} comprises its \emph{laws} together with a \emph{constitution}, each being a point in a metric space, with the constitution specifying the supermajority threshold required to amend it. Members vote to amend the legal corpus by proposing their ideal points, followed by rounds of submitting \emph{public proposals} carrying \emph{supermajority public support}; a polynomial-time \emph{aggregation rule} scores each proposal, and a supported proposal whose score is positive and maximal for two rounds is adopted; if none is found the status quo is retained. 
Public proposals can be sourced from deliberation, vote aggregation, or AI mediation.  
With Constitutional Consensus, a community can run the constitutional governance protocol on members' personal computing devices (e.g., smartphones), achieving digital sovereignty.

We focus on the utility of the generalised median, and study the \emph{compromise gap} between best peak and unconstrained optimum---zero in one dimension, bounded in general, narrowed in simulation by a simple heuristic. We instantiate the framework to seven canonical settings --- electing officers, setting rates, allocating budgets, ranking priorities, selecting boards, drafting bylaws, and amending the constitution.
By drawing on metric-space aggregation, reality-aware social choice, supermajority amendment, constitutional consensus, deliberative coalition formation, and AI mediation, we provide a comprehensive framework for the constitutional governance of digital communities and organisations.

\keywords{Computational social choice \and Constitutional governance \and Metric spaces \and Efficient aggregation \and Deliberation \and Digital cooperatives}
\end{abstract}

\section{Introduction}
\label{section:introduction}

Computational social choice and algorithmic decision theory have produced a rich body of aggregation theory: voting rules, participatory budgeting, committee selection, ranking aggregation, and aggregation over metric spaces. A community wishing to govern itself today, however, still finds no end-to-end process to adopt. Existing contributions tend to address single stages---aggregation, deliberation, amendment, consensus---in isolation, and key aggregators are NP-hard in natural settings. Efficient, well-behaved \emph{integration} of these stages into a coherent governance process is missing.

Concretely, consider a digital community: a freelancers' guild, a neighbourhood association, a regional dog lovers' network. The community must elect officers, set rates, allocate a budget, rank strategic priorities, select a governance board, and draft and amend its own bylaws---each a different decision type, with different structure (unordered alternatives, one-dimensional scales, the simplex, permutations, subsets, strings). Today, a community wishing to do so would have to glue together heterogeneous mechanisms with no unifying contract, while accepting that several of the natural aggregators are computationally intractable.

We propose \emph{constitutional governance in metric spaces}, a unified framework to fill this gap. A community's \emph{legal corpus} comprises its \emph{laws}---substantive components such as rates, budgets, boards, and bylaws---together with a \emph{constitution}. The constitution specifies, per amendable component, including itself, a metric space~$(X, d)$, an aggregation rule~$\phi$, and a supermajority threshold~$\sigma \in [1/2, 1)$ required to amend it. Each member specifies an ideal element, serving as both vote and \emph{personal proposal}; subsequently, any member may submit a \emph{public proposal} from any source---deliberation among members, vote aggregation, or AI mediation---carrying supermajority public support under the revealed votes. Public supermajority support identifies supporters without requiring them to coordinate.
The constitutional governance rule scores each proposal against the status quo via~$\phi$, adopts the supported proposal of positive maximal score, and otherwise retains the status quo. Each component is amended at its own supermajority threshold via the same rule; the constitutional threshold $\sigma_C$ amending itself is the one self-referential case, and uses the $h$-rule of~\cite{abramowitz2021beginning}.

Operationally, the constitutional governance protocol may run on top of constitutional consensus~\cite{keidar2025constitutional} on members' personal computing devices (e.g., smartphones), independently of any global resource, rendering the community digitally sovereign---governed only by the constitution it has itself adopted and amended. The framework is readily applicable to the operation of the assemblies that govern communities and federations as envisioned in Grassroots Federation~\cite{shapiro2025GF}.

The framework integrates six lines of prior work. \emph{Aggregation over metric spaces}~\cite{bulteau2021aggregation} provides ideal elements, distance-induced preferences, a unified treatment across social choice settings, and algorithmic vote aggregation as a source of public proposals. \emph{Reality-aware social choice}~\cite{shapiro2018incorporating} contributes the status quo as an ever-present alternative, sidestepping Plott--McKelvey-style impossibilities. \emph{Supermajority-based constitutional amendment}~\cite{abramowitz2021beginning} contributes the $\sigma$-parameterisation and the self-referential $h$-rule for amending~$\sigma$ itself. \emph{Constitutional consensus}~\cite{keidar2025constitutional} contributes the operational ratification protocol. \emph{Deliberative coalition formation}~\cite{elkind2021united,elkind2022complexity} contributes coalition-derived public proposals. \emph{AI mediation}~\cite{briman2025ai} contributes AI-derived compromise proposals.

The paper is organised as follows. Section~\ref{section:framework} presents the framework; Section~\ref{section:framework_properties} its framework-level properties; Section~\ref{section:rules} develops the generalised median and the strategic properties it confers; Section~\ref{section:metrics} instantiates the seven decision types; Section~\ref{section:gap} defines and studies the compromise gap; Section~\ref{section:discussion} concludes. The appendices contain the proofs, related work, and full details of the settings, amendments, and simulations.

\section{Framework}
\label{section:framework}

We follow the spatial model of preferences~\cite{SpatialModel,bulteau2021aggregation}. A \emph{metric space} $(X, d)$ consists of a set $X$ and a metric $d : X \times X \to \mathbb{R}_{\ge 0}$ satisfying symmetry, $d(x, y) = 0 \iff x = y$, and the triangle inequality. We consider $n$ members. Member~$i$ specifies an \emph{ideal element} $v_i \in X$, interpreted as $i$'s most-preferred element; preferences are induced by distance, so $i$ prefers $x$ to $y$ iff $d(v_i, x) < d(v_i, y)$. We write $V = (v_1, \ldots, v_n)$.

\mypara{The Framework, Informally}\label{subsection:framework-informal} The framework specifies what each amendable component of a community's legal corpus carries and how proposals to amend it are scored. The legal corpus comprises the laws---substantive components such as a current commission rate, a current budget, a current board, and the current bylaws---together with the constitution, which specifies how each component, including itself, is amended. The constitution attaches to each component a metric space~$(X, d)$, an aggregation rule~$\phi$, and a supermajority threshold~$\sigma \in [1/2, 1)$. Each member specifies an ideal element of the space, serving as both vote and initial proposal. In addition, members may submit \emph{public proposals} from any source---own intuition, deliberation among members~\cite{elkind2021united,elkind2022complexity}, vote aggregation~\cite{bulteau2021aggregation}, or AI mediation~\cite{briman2025ai}; the framework specifies admissibility rather than provenance. The constitutional governance rule scores each proposal against the status quo via~$\phi$, adopts the supported proposal of positive maximal score, and otherwise retains the status quo. Each component is amended by a $\sigma_i$-supermajority via the same rule; the constitutional threshold $\sigma_C$ amending itself is the one self-referential case, and uses the $h$-rule of~\cite{abramowitz2021beginning}.

\begin{example}[A Running Example: A Freelancer Cooperative]
\label{subsection:running-example}
Five freelancers---Alice, Bob, Carol, Dan, Eve---found a digital cooperative for matching members with clients. Their initial constitution fixes, for ordinary decisions, $\phi$~=~generalised median and $\sigma = 1/2$; for constitutional amendment decisions, $\sigma = 2/3$.

Their first ordinary decision allocates the budget across marketing~($M$), infrastructure~($I$), and member benefits~($B$). The metric space is the simplex $\Delta^2$ with Euclidean distance; the status quo is $s = (1/3, 1/3, 1/3)$. In the voting round the members submit
\( v_1 = (0.5, 0.3, 0.2), \quad v_2 = (0.5, 0.2, 0.3), \quad v_3 = (0.3, 0.5, 0.2), \quad v_4 = (0.2, 0.4, 0.4), \quad v_5 = (0.2, 0.6, 0.2)\).
In a public-proposal round, a member submits the geometric median of $V$ as a public proposal, computed by an off-the-shelf solver---approximately $\hat c = (0.34, 0.40, 0.27)$.\footnote{The geometric median is the $L_1$ optimum, not the constitutional governance rule; here it serves as a public-proposal candidate, scored under the constitutional generalised median.} The round-2 proposal set is $V \cup \{\hat c\}$, and direct computation shows $\hat c$ attains the highest median utility among supported proposals; the epoch quiesces with $\hat c$ as the new budget.

A subsequent epoch raises the budget threshold from $1/2$ to $3/5$. This amends the budget's $\sigma$ and is decided by the same rule, requiring a $\sigma_C = 2/3$ supermajority. We thread this example through the formalisation that follows.
\end{example}

\mypara{Status Quo, Utility, and Proposal Set}\label{subsection:status-quo} \mysubpara{Status quo} A distinguished element $s \in X$ represents the current value of the component being amended. Following reality-aware social choice~\cite{shapiro2018incorporating}, $s$ is ever-present and evolving: when a proposal is adopted it becomes the new status quo for subsequent decisions. \mysubpara{Utility} The \emph{utility} of a proposal $p \in X$ for a member with ideal element $q$ is
\[ u(q, p) := d(q, s) - d(q, p) , \]
positive iff the member strictly prefers $p$ to $s$. By the triangle inequality $u(q, p) \le d(q, s)$, with equality at $p = q$. \mysubpara{Proposal set} The framework operates iteratively. At any round~$r$ within an epoch, member~$i$ holds an immutable \emph{vote} $v_i \in X$ (sealed at the start of the epoch; Section~\ref{subsection:protocol}) and a current \emph{public proposal} $c_i^r \in X \cup \{\bot\}$, which the member may update across rounds. The \emph{round-$r$ proposal set} is
\( P^r = V \cup \{ c_i^r : c_i^r \neq \bot \}, \qquad n \le |P^r| \le 2n\).
Earlier public proposals that have been overridden no longer appear in $P^r$; the round-$r$ aggregation runs over $P^r$ alone. Note, however, that $\varepsilon$-novelty (Section~\ref{subsection:protocol}) is enforced against \emph{all} previously-submitted public proposals in the epoch, active or overridden, so a withdrawn public proposal cannot be re-submitted unchanged. In the running example, $P^1 = V$ at the end of the voting round, and $P^2 = V \cup \{\hat c\}$ once the geometric-median public proposal is admitted.

\mypara{The Constitutional Governance Rule}\label{subsection:rule} The constitution names, per amendable component, an \emph{aggregation rule} $\phi : \mathbb{R}^n \to \mathbb{R}$ and a \emph{supermajority threshold} $\sigma \in [1/2, 1)$. A proposal $p$ is \emph{supported} at round~$r$ if it has \emph{supermajority public support}: $u(v_i, p) > 0$ for at least $\lceil \sigma n \rceil$ values of $i$. The supporting set $S(p) := \{i : u(v_i, p) > 0\}$ is publicly observable once votes are revealed. The \emph{utility vector} of $p$ is $u(p) = (u(v_1, p), \ldots, u(v_n, p))$, and its \emph{aggregate score} is $u_\phi(p) := \phi(u(p))$.

\begin{definition}[Constitutional governance rule, per round]
\label{definition:rule}
The rule applied at round~$r$ adopts the supported proposal $p^* \in P^r$ with positive maximal $u_\phi$, with tie breaking. If no supported proposal in $P^r$ has positive $u_\phi$, the round has no winner.
\end{definition}

The two roles are kept distinct: $\sigma$ determines which proposals are \emph{admissible} (the supermajority gate), while $\phi$ determines which among them is \emph{adopted} (the maximal score).\footnote{For even~$n$ and $\sigma = 1/2$, the threshold $\lceil \sigma n \rceil = n/2$ represents ``at least half'' rather than a strict majority: a supported proposal needs only $n/2$ members to strictly prefer it, leaving open the possibility of a $1/2$--$1/2$ split. A community that prefers strict-majority semantics for even~$n$ should set $\sigma$ slightly above $1/2$ (so that $\lceil \sigma n \rceil \ge n/2 + 1$). This is a design choice the community makes constitutionally.}

\mypara{The Proposal Protocol}\label{subsection:protocol} The rule runs in \emph{epochs}, each governed by the prevailing constitution. \mysubpara{Voting round} All members simultaneously submit sealed votes; once revealed, votes remain fixed for the duration of the epoch. \mysubpara{Public-proposal rounds} Once votes are revealed, each subsequent round permits members to submit, update, or withdraw their public proposal $c_i^r$, subject to five admissibility conditions: at most one current public proposal per member; \emph{proposer-preference}, $u(v_i, c_i^r) > 0$; \emph{public support}, $c_i^r$ is supported under the now-known votes; \emph{$\varepsilon$-novelty}, $c_i^r$ lies at distance $\ge \varepsilon$ from every vote and every previously-submitted public proposal in the epoch; and \emph{improvement}, $u_\phi(c_i^r)$ strictly exceeds the preceding round's winning score (or zero if there was no winner). \mysubpara{Termination} Definition~\ref{definition:rule} runs at the end of each round. The epoch terminates at \emph{two-round quiescence}: either the same proposal wins two consecutive rounds, or no proposal is supported for two consecutive rounds. The outcome is the quiescent winner, or $s$ if no proposal was supported.

\mypara{Self-Amendment}\label{subsection:self-amendment} Self-amendment is built into the legal corpus.
\begin{tcolorbox}[title=The Legal Corpus and its Amendment]
The legal corpus comprises the laws and the constitution. The constitution attaches to each component~$i$ a metric space $(X_i, d_i)$, an aggregation rule $\phi_i$, a threshold $\sigma_i$, and a novelty distance $\varepsilon_i$. Each component is amended by a $\sigma_i$-supermajority via Definition~\ref{definition:rule}. The constitution itself is a component, with its own threshold $\sigma_C$. The $\sigma_i$ of any component other than the constitution is amended by a $\sigma_C$-supermajority via Definition~\ref{definition:rule}. The constitutional threshold $\sigma_C$ amending itself is the one self-referential case, and it is amended by the $h$-rule of~\cite{abramowitz2021beginning}.
\end{tcolorbox}
The running example's second epoch---raising the budget threshold from $1/2$ to $3/5$---amends the budget's $\sigma$, decided by a $\sigma_C$-supermajority via Definition~\ref{definition:rule}. Full details are deferred to Appendix~\ref{appendix:amendment}.

\section{Framework Properties}
\label{section:framework_properties}

The constitutional governance rule (Definition~\ref{definition:rule}) is parametrised by an aggregation rule~$\phi$. This section collects framework-level properties---existence and complexity, anonymity and neutrality, reality-awareness, Condorcet-cycle immunity---each stated with the minimal hypothesis on~$\phi$ it requires. Any instance (the generalised median of Section~\ref{section:rules}, or any aggregator a community adopts---e.g., the mean in Appendix~\ref{appendix:mean}) inherits these whenever it meets the hypothesis.

\mypara{Complexity} Each round of the constitutional governance rule runs in $O(|P^r| \cdot T_\phi(n))$ time, where $T_\phi(n)$ is the time to evaluate $\phi$ on a vector of length $n$; polynomial whenever $\phi$ is. Per-round complexity for the seven application settings under the generalised median is given in Table~\ref{table:applications}.

\mypara{Finite Termination}\label{subsection:termination} The proposal protocol of Section~\ref{subsection:protocol} runs in epochs of indefinite length, but every epoch terminates after finitely many rounds. We state the guarantee under a mild structural hypothesis.

\begin{definition}[Totally bounded]
A metric space $(X, d)$ is \emph{totally bounded} if for every $\varepsilon > 0$ and every bounded $B \subseteq X$, $B$ admits a finite cover by $\varepsilon$-balls.
\end{definition}

\begin{restatable}{proposition}{proptermination}
\label{proposition:termination}
Fix a constitutionally-specified novelty distance $\varepsilon > 0$. If $(X, d)$ is totally bounded, every epoch of the proposal protocol terminates after finitely many public-proposal rounds.
\end{restatable}

Proof and verification that the seven application settings of Section~\ref{section:metrics} are totally bounded are in Appendix~\ref{appendix:termination-proof}.

\mypara{Anonymity and Neutrality} These are treated in Appendix~\ref{appendix:anonymity}.

\mypara{Reality-Awareness} Following reality-aware social choice~\cite{shapiro2018incorporating}, the framework treats the status quo as an ever-present alternative that is retained unless a proposal genuinely improves on it. By Definition~\ref{definition:rule}, if $\phi(\mathbf{0}) = 0$ then $u_\phi(s) = 0$, so the constitutional governance rule retains the status quo unless some supported proposal has $u_\phi > 0$. The generalised median satisfies $\phi(\mathbf{0}) = 0$, as does any anonymous, monotone aggregator that agrees with indifference at zero.

\mypara{Condorcet-Cycle Immunity} Classical majority rule over three or more alternatives can cycle: a majority may prefer $a$ to $b$, $b$ to $c$, and $c$ to $a$. In one-dimensional settings with single-peaked preferences, Black's median voter theorem~\cite{blackmedian} guarantees a Condorcet winner. In multidimensional settings, however, Plott~\cite{plott1967notion} showed that a Condorcet winner exists only under a radial-symmetry condition that generically fails, and McKelvey~\cite{mckelvey1976intransitivities} showed that when it fails, the top cycle of majority rule spans the entire space. The constitutional governance rule reduces every comparison to a scalar score against the fixed status quo $s$, so no pairwise majority structure arises and no cycling is possible.

\begin{restatable}{theorem}{thmcondorcet}
\label{theorem:condorcet}
For any aggregation rule~$\phi$, any threshold $\sigma \in [1/2, 1)$, and any metric space $(X, d)$ of any dimension or structure: at every round~$r$, either a supported proposal with positive maximal~$u_\phi$ is adopted, or the round has no winner. No cycling among proposals is possible.
\end{restatable}

Plott--McKelvey concern pairwise majority comparisons among alternatives, which the constitutional governance rule never forms: each proposal is scored against the fixed status quo $s$ alone, yielding a real-valued score that is well-defined in any number of dimensions. In one dimension under the generalised median, the rule's winner is the positional median (Proposition~\ref{proposition:1d-characterisation}), which by Black's theorem~\cite{blackmedian} is the Condorcet winner whenever one exists.

\mypara{Strategy-Proofness} Strategy-proofness depends on $\phi$; for the generalised median at majority threshold, no misreport weakly dominates sincere voting under the proposal protocol (Theorem~\ref{theorem:protocol-sp}, Section~\ref{section:rules}).

\mypara{Monotonicity}\label{subsection:monotonicity} Monotonicity is treated in Appendix~\ref{appendix:monotonicity-counterexample}.

\section{The Generalised Median}
\label{section:rules}

The framework is parametrised by an aggregation rule~$\phi$. We develop the generalised median as our worked aggregation rule, and study the strategic properties it confers on the constitutional governance rule. The generalised median satisfies the framework-level hypotheses of Section~\ref{section:framework_properties}: $O(n)$-computable, symmetric, and zero at indifference. The mean is treated as an alternative utilitarian instantiation in Appendix~\ref{appendix:mean}.

\begin{definition}[Generalised median]
\label{definition:median}
For $\sigma \in [1/2, 1)$, the \emph{generalised median at threshold $\sigma$} is
\( \phi_\sigma(u) := \text{the } \lceil \sigma n \rceil\text{-th largest entry of } u_1, \ldots, u_n\).
\end{definition}

For $\sigma = 1/2$ and odd $n$, $\phi_\sigma$ coincides with the ordinary median. The generalised median has a useful equivalence with support: $\phi_\sigma(u(p)) > 0$ if and only if at least $\lceil \sigma n \rceil$ entries of $u(p)$ are strictly positive, that is, if and only if $p$ is supported. Under the generalised median instance, the supermajority gate and the positivity of $u_\phi$ \emph{coincide}, so the rule reduces to: adopt the proposal with maximal $\phi_\sigma(u(p))$ if any has $\phi_\sigma > 0$, else retain the status quo. Note that personal proposals are subject to the same gate as public proposals: a personal proposal $v_i$ is supported only if at least $\lceil \sigma n \rceil$ members strictly prefer it to the status quo under their revealed votes.

The generalised median is motivated by cooperative governance and robustness. Economic models of cooperative behaviour require the cooperative's decision process to be representable as the maximisation of an objective function~\cite{hart1996governance}; for $\sigma = 1/2$, the rule maximises the utility of the median member. The median is also robust to outliers: a member with an extreme ideal element shifts the value of $\phi_\sigma$ only insofar as their utility crosses the $\lceil \sigma n \rceil$-th position, not in proportion to the magnitude of their utility.

\begin{remark}[The median beyond one dimension]
\label{remark:median-spaces}
In 1D, Black's median voter theorem justifies the median directly. Beyond 1D, the classical theorem does not apply, but the median retains a privileged status on \emph{median spaces}~\cite{nehring2007structure}---trees, hypercubes, products of lines---as the unique strategy-proof, anonymous, efficient rule. For non-median spaces, the framework sidesteps the multidimensional impossibilities entirely via the status quo (Theorem~\ref{theorem:condorcet}), and cooperative governance provides an independent justification: economic models require maximisation of an objective function, which the median utility supplies.
\end{remark}

\mypara{Majoritarity}\label{subsection:majoritarity} For $\sigma > 1/2$, the median follows a majority. 
\begin{definition}[$\sigma$-majoritarity]
\label{definition:majoritarity}
An aggregation method is \emph{$\sigma$-majoritarian} if, whenever at least $\lceil \sigma n \rceil$ members share an ideal element $w$ that is preferred to $s$, the rule adopts $w$.
\end{definition}

\begin{restatable}{proposition}{propmajoritarian}
\label{proposition:median-majoritarian}
For $\sigma > 1/2$ or odd $n$, the constitutional governance rule with $\phi = \phi_\sigma$ is $\sigma$-majoritarian.
\end{restatable}

\mypara{Per-Round Strategic Behaviour}\label{subsection:per-round-sp} We turn to strategic incentives. We first characterise per-round behaviour, then prove the protocol-level result (Section~\ref{subsection:protocol-sp}).

\begin{restatable}[1D dominance, generalised median]{proposition}{propdominance}
\label{proposition:1d-dominance}
In one-dimensional settings with $X \subseteq \mathbb{R}$ and $d(x, y) = |x - y|$, sincere voting is a weakly dominant strategy under $\phi_\sigma$ at any single round, regardless of $\sigma \in [1/2, 1)$ and regardless of public proposals submitted by other members.
\end{restatable}

A worked example of 1D dominance is in Appendix~\ref{appendix:1d-characterisation}.
In dimensions two and above, the result fails: a single misreporter can manipulate the generalised median over utilities.

\begin{restatable}[Multidimensional non-SP, generalised median]{proposition}{propmultidimnotsp}
\label{proposition:multidim-not-sp}
For $X = \mathbb{R}^2$ with Euclidean distance and $\phi_\sigma$ with $\sigma = 1/2$, the constitutional governance rule per round is not strategy-proof.
\end{restatable}

The per-round structure of manipulation is characterised by three lemmas in Appendix~\ref{appendix:per-round-anatomy}; the resulting dichotomy motivates the protocol design of Section~\ref{subsection:protocol-sp}.

\mypara{Strategy-Proofness Under the Proposal Protocol}\label{subsection:protocol-sp} The proposal protocol of Section~\ref{subsection:protocol} seals votes at the start of the epoch and fixes them throughout, denying the misreporter both the information needed to construct an informed misreport and the option to revise mid-epoch. Under the protocol (and with no information leakage), no misreport weakly dominates sincere voting at majority threshold: for every misreport there is at least one profile of other members' votes under which sincere strictly outperforms it.

\begin{restatable}[Ex-post separability of sincere voting under the protocol]{theorem}{thmprotocolsp}
\label{theorem:protocol-sp}
For $\sigma = 1/2$ and $\phi = \phi_\sigma$, under the proposal protocol, no misreport weakly dominates sincere voting: for every member~$i$, every true ideal $v_i^* \neq s$, and every misreport $\tilde v_i \neq v_i^*$, there exists an epoch (a profile of the other members' votes and an admissible run) in which member~$i$'s true utility for the outcome is strictly higher under sincere voting than under the misreport.
\end{restatable}

\begin{restatable}[1D dominance under the protocol]{corollary}{coronedprotocolsp}
\label{corollary:1d-protocol-sp}
In one-dimensional settings, sincere voting is a weakly dominant strategy under $\phi_\sigma$ for any $\sigma \in [1/2, 1)$ under the proposal protocol.
\end{restatable}

The case $\sigma > 1/2$ in general metric spaces is open.

\mypara{Public-Proposal and Coalition Behaviour}\label{subsection:compromise-behaviour} Public-proposal and coalition-channel behaviour is characterised in Appendix~\ref{appendix:coalition}.

\section{Decision Types}
\label{section:metrics}

Each amendable component has its own metric space. We adopt the metric-space instantiations of Bulteau et al.~\cite{bulteau2021aggregation} for the seven canonical settings of cooperative governance: electing officers, setting rates, allocating a budget, ranking priorities, selecting a board, drafting bylaws, and amending the constitution.

\mypara{Overview}\label{subsection:overview} Table~\ref{table:applications} summarises the seven settings: the metric space and distance function, the per-round complexity of the constitutional governance rule, and the existence and complexity of Bulteau et al.'s $L_p$ aggregation---selecting $x \in X$ minimising $\sum_i d(x, v_i)^p$---in the same setting. The constitutional governance rule is polynomial-time in every setting, including those for which $L_p$ aggregation is NP-hard, because the rule scores each proposal in $P^r$ rather than searching the entire metric space.

\begin{table}[t]
\caption{The seven application settings: metric space, per-round complexity of the constitutional governance rule, and Bulteau et al.'s $L_p$ aggregation. $n$~=~members, $m$~=~alternatives, $|A|$~=~candidates, $\ell$~=~text length.}
\label{table:applications}
\centering
\small
\begin{tabular}{@{}llll@{}}
\toprule
Setting & Metric space & Per round & $L_p$~\cite{bulteau2021aggregation} \\
\midrule
Plurality (officers) & Discrete on $A \cup \{\bot\}$ & $O(n^2)$ & linear \\
1D elections (rates) & $X \subseteq \mathbb{R}$, $|x-y|$ & $O(n^2)$ & linear \\
Simplex (budgeting) & $\Delta^{m-1}$, Euclidean & $O(n^2 m)$ & efficient \\
SWF (rankings) & Permutations $\mathbb{S}_A$, swap distance & $O(n^2 m \log m)$ & NP-hard \\
Committee (boards) & Subsets $2^A$, symmetric difference & $O(n^2 |A|)$ & NP-hard ($p > 1$) \\
Legislation (bylaws) & Strings $\Sigma^*$, weighted Levenshtein & $O(n^2 \ell^2)$ & NP-hard \\
Amendments & Per component (see \S\ref{subsection:self-amendment}) & polynomial & --- \\
\bottomrule
\end{tabular}
\end{table}

We work two settings in detail in the main text---1D rate-setting and simplex budgeting, the latter being the running example of Example~\ref{subsection:running-example}. The other five settings receive compressed treatment below, with full details in Appendix~\ref{appendix:settings}.

\mypara{Worked Setting: Setting a Rate}\label{subsection:rate} A community must set various parameters that take values on a one-dimensional scale: a commission rate, a membership fee, a minimum hourly wage, a fraction of revenue allocated to reserves. Each member proposes a value; the current value is the status quo. \mysubpara{Model} The metric space is $X \subseteq \mathbb{R}$ with $d(x, y) = |x - y|$. The status quo $s \in X$ is the current value. Each member~$i$ specifies an ideal value $v_i \in X$. The constitutional threshold $\sigma$ is constitutionally specified per parameter. The 1D structure admits a sharp characterisation of the generalised median's per-round outcome; see Appendix~\ref{appendix:1d-characterisation}. Proposition~\ref{proposition:1d-characterisation} drives the per-round strategy-proofness result of Proposition~\ref{proposition:1d-dominance}: in 1D, the winner depends only on the positional voter $v_m$, which a single misreport can move only weakly away from the misreporter's ideal element.

\begin{example}[Setting a commission rate]
\label{example:rate}
A cooperative of five members sets the commission rate at $\sigma = 1/2$. The current rate is $s = 20\%$. Members propose $v_1 = 10\%, v_2 = 15\%, v_3 = 18\%, v_4 = 22\%, v_5 = 25\%$; the median is $v_m = v_3 = 18\%$. 
The median utility of each proposal:
\begin{center}
\begin{tabular}{@{}lccccc@{}}
\toprule
$p$ & 10\% & 15\% & 18\% & 22\% & 25\% \\
$\phi_{1/2}(u(p))$ & $-6$ & $-1$ & $+2$ & $-2$ & $-5$ \\
\bottomrule
\end{tabular}
\end{center}
Only $p = 18\%$ is supported, and it has positive median utility, so the rate is amended to $18\%$.
\end{example}

\mypara{Worked Setting: Allocating a Budget}\label{subsection:budgeting} A community must allocate its budget across categories: marketing, infrastructure, member benefits, a reserve fund. Each member proposes a distribution; the current allocation is the status quo. This is the running example of Example~\ref{subsection:running-example}. \mysubpara{Model} Let $A$ be the set of categories, $m = |A|$. The metric space is the simplex $\Delta^{m-1} = \{(w^1, \ldots, w^m) \in \mathbb{R}_{\ge 0}^m : \sum_i w^i = 1\}$ with the Euclidean metric $d(x, y) = \|x - y\|_2$. The status quo $s \in \Delta^{m-1}$ is the current allocation. Each member~$i$ specifies an ideal allocation $v_i \in \Delta^{m-1}$. \mysubpara{Round complexity} For each $p \in P^r$, computing the utility vector takes $O(n m)$ time, evaluating $\phi_\sigma$ takes $O(n)$ time. Over $|P^r| \le 2n$ proposals, the per-round cost is $O(n^2 m)$. \mysubpara{Comparison with $L_1$ over the simplex} Bulteau et al.'s $L_1$ aggregation over the simplex returns the geometric median of the votes, which is computable in $\tilde O(nm)$ time. The geometric median is in general \emph{not} an element of $V$, but the constitutional governance rule restricts to $P^r$. Under the running example, the geometric median was admitted as a public proposal (Example~\ref{subsection:running-example}) and won the round; this illustrates the open public-proposal channel admitting the optimization-based proposal alongside the peaks. The rule scores it under $\phi_\sigma$, not under the rule used to compute it. \mysubpara{Experimental support for median rules on the simplex} Puppe and Rollmann~\cite{puppe2021mean} compared mean and median voting rules on the simplex in a laboratory experiment, finding that voters vote more sincerely under the median rule and more strategically under the mean.

\mypara{Plurality (electing officers)}\label{subsection:other-settings} The five remaining settings receive compact treatment here, with full details in Appendix~\ref{appendix:settings} (Appendix~\ref{appendix:amendment} for amendments). The metric space is $X = A \cup \{\bot\}$ where $A$ is the candidate set and $\bot$ denotes the vacant position; $d$ is the discrete metric. With $s = \bot$ as the status quo, member utilities are binary: $1$ if the proposal matches the member's preferred candidate, $0$ otherwise. The rule selects the candidate with the most supporters, provided that candidate has $\sigma$-supermajority support; otherwise the position remains vacant.

\mypara{Social welfare functions (ranking priorities)} The metric space is $X = \mathbb{S}_A$, the set of permutations over $A$, with the swap distance (number of adjacent transpositions). Bulteau et al.'s $L_p$ aggregation here corresponds to the Kemeny ranking, which is NP-hard. The constitutional governance rule selects the supported ranking with maximal score in $O(n^2 m \log m)$ time, trading optimality over the entire space for tractability over the proposal set.

\mypara{Committee elections (selecting a board)} The metric space is $X \subseteq 2^A$ (subsets of eligible members), with the symmetric difference $d(x, y) = |x \triangle y|$. Bulteau et al.'s Condorcet aggregation is coNP-hard and $L_p$ for $p > 1$ is NP-hard. The constitutional governance rule selects the supported board with maximal score in $O(n^2 |A|)$ time. With $s = \emptyset$ as the status quo (no board), member utilities measure how many of their preferred members the proposed board includes.

\mypara{Legislation (drafting bylaws)} A community must draft and amend the bylaws governing membership, decision-making, and dispute resolution. This is the most challenging setting in Bulteau et al.~\cite{bulteau2021aggregation}: $L_p$ aggregation is NP-hard for all $p \ge 1$. The constitutional governance rule applies regardless: each member submits a text $v_i \in \Sigma^*$, the rule scores each proposal under the metric, and the supported text with maximal score is adopted in $O(n^2 \ell^2)$ time, where $\ell$ is the length of the longest proposed text.

\mypara{Constitutional amendments} Each amendable component carries its own metric space, aggregation rule, threshold $\sigma_i$, and novelty distance, and is amended by a $\sigma_i$-supermajority via Definition~\ref{definition:rule}. The constitutional threshold $\sigma_C$ amending itself is the one self-referential case, and it is amended by the $h$-rule of~\cite{abramowitz2021beginning}. Section~\ref{subsection:self-amendment} states this; full details are in Appendix~\ref{appendix:amendment}.

\begin{remark}[Semantic metrics]\label{remark:semantic-metrics}
A note on the metric for legislation. Bulteau et al.\ adopt a weighted Levenshtein distance, a syntactic measure that counts character-level edit operations without regard to meaning. However, two bylaws differing by a synonym substitution are syntactically distant but semantically identical; two differing by a single negation are syntactically close but semantically opposite. Thus, what governance requires is a \emph{semantic} distance between texts. The constitutional governance rule is metric-agnostic: any distance function over $\Sigma^*$ plugs into the framework with no change to the rule, the protocol, or any framework-level guarantee. AI mediation~\cite{briman2025ai} and large language models can compute such distances, which the framework can use directly. We retain the syntactic Levenshtein distance here for comparability with Bulteau et al.; the design of governance-appropriate semantic distances is an active, orthogonal research direction.
\end{remark}

\section{The Compromise Gap: An Empirical Study}
\label{section:gap}

The constitutional governance rule selects an outcome from the proposal set $P^r$ rather than from the full metric space $X$. This trades the unconstrained optimum for a polynomial-time process a community can actually run. We quantify the trade-off via the \emph{compromise gap}: the difference between the best peak and the unconstrained optimum. The compromise gap is what public proposals \emph{close}: any public proposal from any source---a member's intuition, a coalition's deliberation, an optimization algorithm, or an AI mediator---narrows the gap when it scores higher than the best peak.

\mypara{Definition and Scope}\label{subsection:gap-definition} Fix a metric space $(X, d)$, a vote profile $V = \{v_1, \ldots, v_n\}$, a status quo $s$, an aggregation rule $\phi$, and a threshold $\sigma$. Let
\[
\text{OPT}_\phi(V, s) := \sup_{x \in X} \phi(u(x)), \qquad \text{PEAK}_\phi(V, s) := \max_{p \in V} \phi(u(p))
\]
be the unconstrained optimum and the best peak. The \emph{compromise gap} is the non-negative quantity
\[
\text{CG}_\phi(V, s) := \text{OPT}_\phi(V, s) - \text{PEAK}_\phi(V, s) \ge 0 .
\]
We restrict attention to profiles with $\text{OPT}_\phi(V, s) > 0$, since otherwise the status quo is retained under both the unconstrained and the proposal-restricted rules and the comparison is vacuous. A public proposal $c \in X$ \emph{closes the gap} (in part or whole) when $\phi(u(c)) > \text{PEAK}_\phi(V, s)$. The gap is fully closed if the proposal set contains an element attaining $\text{OPT}_\phi$. The 1D result below is specific to the generalised median; the Lipschitz bound applies to any 1-Lipschitz $\phi$ and any $(X, d)$; the simulations use the generalised median.

\mypara{Zero Gap in One Dimension}\label{subsection:zero-gap} The first result is a clean positive: in one dimension with the generalised median, peaks suffice---the gap is zero.

\begin{restatable}[Zero gap in 1D]{theorem}{thmzerogap}
\label{theorem:zero-gap-1d}
Let $X \subseteq \mathbb{R}$ with $d(x, y) = |x - y|$, let $n$ be odd, and let $\sigma = 1/2$. Then $\text{CG}_{\phi_\sigma}(V, s) = 0$ for every $V \subseteq X$ and every $s \in X$.
\end{restatable}

Theorem~\ref{theorem:zero-gap-1d} is the median voter theorem expressed in our framework: in 1D with odd electorate, the median peak is both the unconstrained optimum and a peak. The result extends to the constitutional scalar parameters of Appendix~\ref{appendix:scalar} (thresholds and novelty distances) under the same argument.

\mypara{General Bound}\label{subsection:lipschitz} In dimensions two and above, the gap can be positive; minimal examples are in Appendix~\ref{appendix:gap-examples}. The phenomenon is intrinsic to multidimensional settings: the unconstrained optimum need not coincide with any peak. The following bound holds in full generality.

\begin{restatable}[Lipschitz bound]{proposition}{proplipschitz}
\label{proposition:lipschitz}
If $\phi : \mathbb{R}^n \to \mathbb{R}$ is $1$-Lipschitz in the $\ell_\infty$ norm, then
\(
\text{CG}_\phi(V, s) \le \min_{p \in V} d(x^*, p)\), where \(x^* \in \arg\max_{x \in X} \phi(u(x))\).
\end{restatable}

Both the generalised median and the mean are $1$-Lipschitz in $\ell_\infty$. For the generalised median, this follows from the fact that the sort operator is $1$-Lipschitz in $\ell_\infty$, so the order statistic inherits the property. For the mean, $|\frac{1}{n}\sum_i (u_i - u'_i)| \le \frac{1}{n}\sum_i |u_i - u'_i| \le \|u - u'\|_\infty$. The bound is tight in Example~\ref{example:star}, where $d(h, V) = 1 = \text{CG}$.

\begin{remark}
Sharper bounds than Proposition~\ref{proposition:lipschitz} would exploit the structure of the metric space and the aggregation rule. Two directions are natural. First, instance-dependent bounds: in median spaces~\cite{nehring2007structure} (trees, hypercubes, products of lines), generalised medians have additional structure that may yield tighter analytic bounds. Second, smoothed-analysis bounds: under realistic priors on $V$ (e.g., uniform sampling, log-concave distributions), expected $\text{CG}$ may admit polynomial-of-$n$ rates rather than worst-case bounds. Both directions are open.
\end{remark}

We turn to algorithmic gap-closing.

\mypara{Closing the Gap: Pairwise Compromise}\label{subsection:heuristic} The framework's open public-proposal channel admits public proposals from any source. We study a simple, polynomial-time source---pairwise proposal combinations---and show empirically that it closes a substantial fraction of the gap in realistic configurations. The formal definition of Heuristic P and its complexity are in Appendix~\ref{appendix:heuristic-details}. Heuristic P formalises the simplest public-proposal move: blend two existing proposals into a candidate, score it under the framework's own rule. The gating step ensures P never recommends a public proposal strictly worse than the current best. The heuristic is proposal-native: input, output, and search space all live in the proposal set or its pairwise combinations. Heuristic P carries no worst-case approximation guarantee. Configurations in which $x^*$ lies far from every pairwise combination of existing proposals are not addressed by P. We turn to empirical evidence.

\mypara{Simulation: Closing the Gap in Practice}\label{subsection:simulation} We simulated the compromise gap and Heuristic P across four of the application settings: 2D Euclidean (a stylised common case), the simplex under Euclidean distance, the hypercube under symmetric difference (Appendix~\ref{appendix:committee}), and permutations under swap distance (Appendix~\ref{appendix:swf}). For each profile, we compute $\text{OPT}$ exactly (by enumeration in finite spaces, by grid search elsewhere), compute $\text{PEAK}$, and run Heuristic P. We report three statistics per configuration:
\begin{itemize}
    \item \emph{Positive-CG frequency}: the fraction of profiles with $\text{CG} > 0$, when $\text{OPT} > 0$.
    \item \emph{Gap-closing ratio}: the mean of $(\phi(u(c^*)) - \text{PEAK}) / \text{CG}$ over profiles with $\text{CG} > 0$ and Heuristic P returning $c^* \neq \bot$.
    \item \emph{Hit rate}: the fraction of profiles on which Heuristic P returns a non-$\bot$ output.
\end{itemize}
Profiles are generated by sampling peaks uniformly: in 2D Euclidean, peaks are drawn uniformly from a bounded region; on the simplex, from the uniform Dirichlet distribution; on the hypercube, uniformly from $2^A$; on permutations, uniformly from $\mathbb{S}_A$. Sample sizes: 200--500 profiles per configuration. Table~\ref{table:simulation} reports nine representative rows; the full sweep across all twenty-seven configurations is in Appendix~\ref{appendix:simulation-full}.

\begin{table}[h]
\caption{Compromise gap and Heuristic P across four settings. ``Gap-closing'' = mean fraction of $\text{CG}$ closed by Heuristic P on profiles with $\text{CG} > 0$ on which P returns a non-$\bot$ output; ``Hit rate'' = fraction of profiles on which P returns a non-$\bot$ output; ``Combined'' = their product, the effective gap closure over all profiles with $\text{CG} > 0$.}
\label{table:simulation}
\centering
\small
\begin{tabular}{@{}llrrrr@{}}
\toprule
Setting & $n$ & CG~$> 0$ & Gap-closing & Hit rate & Combined \\
\midrule
2D Euclidean & 5 & 99\% & 41\% & 67\% & 27\% \\
2D Euclidean & 21 & 100\% & 79\% & 96\% & 76\% \\
2D Euclidean & 51 & 96\% & 96\% & 98\% & 94\% \\
Simplex ($m=4$) & 5 & 97\% & 52\% & 78\% & 41\% \\
Simplex ($m=4$) & 21 & 99\% & 81\% & 96\% & 78\% \\
Hypercube ($|A|=8$) & 5 & 69\% & 97\% & 66\% & 64\% \\
Hypercube ($|A|=8$) & 21 & 36\% & 100\% & 36\% & 36\% \\
Permutations ($m=5$) & 5 & 63\% & 82\% & 53\% & 43\% \\
Permutations ($m=5$) & 21 & 32\% & 99\% & 31\% & 31\% \\
\bottomrule
\end{tabular}
\end{table}

Heuristic P closes a growing fraction of CG as $n$ grows; the effective rate is given in the Combined column.

\section{Outlook}
\label{section:discussion}

We have presented a framework that integrates aggregation over metric spaces, reality-aware social choice, supermajority-based constitutional amendment, constitutional consensus, deliberative coalition formation, and AI mediation into a per-round polynomial-time process; we close with implications and open questions.

\mypara{What the framework provides} A constitutional governance rule that scales across decision types (officers, rates, budgets, rankings, boards, bylaws, the constitution), a supermajority gate, and an open channel admitting proposals from any source. Constitutional design reduces to specifying, per component, the metric space, aggregation rule, threshold, and novelty distance.

\mypara{Cooperative governance and the choice of $\phi$} The generalised median is motivated by cooperative governance: economic models require the cooperative's decision process to be representable as the maximisation of an objective function, and for $\sigma = 1/2$ the rule maximises the median member's utility. A community whose objective is utilitarian would select the mean instance (Appendix~\ref{appendix:mean}). The choice is itself a constitutional decision.

\mypara{Public proposals from any source} The framework admits proposals from any source---pairwise heuristics, optimization, AI mediation, coalition deliberation---on equal footing. In NP-hard settings, approximation algorithms and AI mediators are first-class participants: an approximate optimum that scores above the best peak wins.

\mypara{Four future work directions}   \mysubpara{Sharper compromise gap bounds} The Lipschitz bound of Proposition~\ref{proposition:lipschitz} is loose. Two directions are open: instance-dependent bounds exploiting median-space structure, and smoothed-analysis bounds under realistic priors on~$V$. \mysubpara{Cooperative economics} The connection to economic models of cooperatives should be developed, incorporating heterogeneous member preferences, multi-dimensional decision spaces, and the economic consequences of generalised median versus mean instances. \mysubpara{Implementation and empirical evaluation} The framework is concrete enough to deploy. Empirical evaluation in realistic cooperative settings---deliberation dynamics, round cadence choices, the interplay between vote updates and public-proposal submissions over time---would complement the formal analysis with operational evidence. \mysubpara{Architectural context} The framework can serve as a decision layer for grassroots digital communities~\cite{shapiro2023grassrootsBA}, alongside the consensus layer of~\cite{keidar2025constitutional} and the federation layer of~\cite{shapiro2025GF}. Composing the three layers into a complete grassroots-governance architecture is a natural direction.

\mypara{Acknowledgement} This work is an answer to Francesco Caselli's request---a decision rule that maximises the utility of the median worker---needed to develop the economics of an incentive-compatible digital cooperative. We now hope and expect such economics to be developed.

\bibliographystyle{splncs04}
\bibliography{bib}

\newpage

\appendix

\section{Proofs}
\label{appendix:proofs}

This appendix collects the proofs of results stated in the body, in document order.

\proptermination*
\begin{proof}
We bound the total number of distinct public proposals that can be admitted in a single epoch, which suffices: once no further admissible public proposal exists, two-round quiescence (no public proposal admitted, hence the same proposal set and same winner) is reached in at most one additional round.

Fix an epoch with sealed votes $V = \{v_1, \ldots, v_n\}$. By the proposer-preference condition, every admissible public proposal $c$ satisfies $u(v_i, c) > 0$ for the proposer~$i$, i.e., $d(v_i, c) < d(v_i, s)$. Hence every admissible public proposal lies in
\[ B := \bigcup_{i=1}^n B_{d(v_i, s)}(v_i), \]
the union of open balls around the votes with radii equal to their distances from the status quo. The set $B$ is bounded; since $(X, d)$ is totally bounded, $B$ admits a finite cover by $\varepsilon/2$-balls. Let $N$ denote the size of such a cover.

The $\varepsilon$-novelty condition requires that every admitted public proposal lie at distance $\ge \varepsilon$ from every previously-submitted public proposal (active or overridden) and from every vote. Two public proposals in the same $\varepsilon/2$-ball are at distance $< \varepsilon$ from each other by the triangle inequality, so each $\varepsilon/2$-ball contains at most one admitted public proposal across the entire epoch. Hence the number of distinct public proposals admitted in the epoch is at most $N$, finite.

When no further admissible public proposals exist the round produces the same winner (or no winner) as the previous round. Two-round quiescence is reached, and the epoch terminates. \qed
\end{proof}

\thmcondorcet*
\begin{proof}
Each proposal $p \in P^r$ receives a real-valued aggregate score $u_\phi(p) \in \mathbb{R}$. The maximum of a finite set of real numbers always exists, so when at least one supported proposal has $u_\phi > 0$, the set of supported proposals attaining the maximum is non-empty and the constitutionally-specified tie-break selects one for adoption. Otherwise no proposal is adopted. Since each proposal is compared only against the status quo (yielding a real number) and never against another proposal, no pairwise majority comparison enters the rule, and the strict total order on real numbers precludes cycling. \qed
\end{proof}

\propmajoritarian*
\begin{proof}
Suppose at least $\lceil \sigma n \rceil$ members share ideal element $w \neq s$. For the proposal $p = w$, every member with $v_i = w$ has $u(v_i, w) = d(w, s)$, the maximum possible utility (by the triangle inequality, $u(v_i, p) \le d(v_i, s)$ with equality at $p = v_i$). The $\lceil \sigma n \rceil$-th largest entry of $u(w)$ is therefore $d(w, s)$, so $\phi_\sigma(u(w)) = d(w, s) > 0$ and $w$ is supported.

For any $p' \neq w$, each of the $\lceil \sigma n \rceil$ members at $w$ has $u(w, p') = d(w, s) - d(w, p') < d(w, s)$. So at least $\lceil \sigma n \rceil$ entries of $u(p')$ are strictly below $d(w, s)$, leaving at most $n - \lceil \sigma n \rceil$ entries that could attain $d(w, s)$. For $\sigma > 1/2$ or odd $n$, $n - \lceil \sigma n \rceil < \lceil \sigma n \rceil$, so the $\lceil \sigma n \rceil$-th largest entry of $u(p')$ is strictly below $d(w, s)$. Hence $\phi_\sigma(u(p')) < \phi_\sigma(u(w))$, and $w$ is the unique winner. \qed
\end{proof}

\propdominance*
\begin{proof}
Order the votes as $v_{(1)} \le \cdots \le v_{(n)}$, and let $k = \lceil \sigma n \rceil$. In 1D, the utility $u(v_q, p) = |v_q - s| - |v_q - p|$ is monotone in $v_q$ for each fixed $p$ (non-decreasing when $p > s$, non-increasing when $p < s$); hence the $k$-th largest entry of $u(p)$ equals $u(v_m, p)$ for $v_m := v_{(k)}$ (when $p > s$) or $v_{(n-k+1)}$ (when $p < s$). Thus $\phi_\sigma(u(p)) = |v_m - s| - |v_m - p| \le |v_m - s|$, with equality only at $p = v_m$; the winning proposal in $V$ is the one closest to $v_m$, and no public proposal alters this winner.

Fix member $i$ with true ideal $v_i^* \in \mathbb{R}$. Let $v_m^0$ be the $k$-th positional vote under sincere reporting and $v_m'$ under any misreport $\tilde v_i$. The member's true utility for the winner is $|v_i^* - s| - |v_i^* - v_m'|$, decreasing in $|v_i^* - v_m'|$. We show $|v_i^* - v_m^0| \le |v_i^* - v_m'|$ for every $\tilde v_i$.

If $v_i^* > v_m^0$, any misreport $\tilde v_i \ge v_m^0$ leaves $v_m^0$ unchanged; any misreport $\tilde v_i < v_m^0$ shifts the $k$-th positional vote leftward, giving $v_m' \le v_m^0 < v_i^*$ and $|v_i^* - v_m'| \ge |v_i^* - v_m^0|$. The case $v_i^* < v_m^0$ is symmetric. If $v_i^* = v_m^0$, any misreport either leaves $v_m^0$ fixed or shifts to a neighbouring vote, weakly farther from $v_i^*$. \qed
\end{proof}

\propmultidimnotsp*
\begin{proof}
Take $n = 3$, $s = (0,0)$, true ideals $v_1^* = (1,0)$, $v_2^* = (0,1)$, $v_3^* = (-1,-1)$. Under sincere voting with no public proposals, no proposal is $1/2$-supported with positive median utility, so the status quo is retained and member~1's true utility is~0. Under misreport $\tilde v_1 = (0.5, 0.5)$, the proposal $(0.5, 0.5)$ has utilities approximately $(0.707, 0.293, -0.707)$ on reported votes, with median $0.293 > 0$ and two positive entries; $(0.5, 0.5)$ is supported and is the winner. Member~1's true utility for it is $1 - \sqrt{0.5} \approx 0.293 > 0$. \qed
\end{proof}

\thmprotocolsp*
\begin{proof}
Fix member~$i$, true ideal $v_i^* \neq s$, and misreport $\tilde v_i \neq v_i^*$; let $k = \lceil n/2 \rceil$. We construct a separating profile: $k-1$ of the other members vote $v_i^*$, $k-1$ vote $\tilde v_i$, and the remaining $n - 2k + 1$ vote $s$. (For odd $n$, $n - 2k + 1 = 0$; for even $n$, exactly one member votes $s$.) No public proposals are submitted in any round, so the analysis reduces to the voting round.

\emph{Sincere case.} Member~$i$ votes $v_i^*$. Total votes: $k$ at $v_i^*$, $k-1$ at $\tilde v_i$, $n - 2k + 1$ at $s$.
\begin{itemize}
    \item Proposal $v_i^*$: the $k$ voters at $v_i^*$ each have utility $d(v_i^*, s)$; the $k-1$ voters at $\tilde v_i$ each have utility $d(\tilde v_i, s) - d(\tilde v_i, v_i^*) \le d(v_i^*, s)$ (triangle inequality); the voters at $s$ have utility $-d(v_i^*, s)$. The $k$-th largest is $d(v_i^*, s)$, so $\phi_\sigma(u(v_i^*)) = d(v_i^*, s) > 0$ and $v_i^*$ is supported.
    \item Proposal $\tilde v_i$: the $k-1$ voters at $\tilde v_i$ have utility $d(\tilde v_i, s)$; the $k$ voters at $v_i^*$ have utility $d(v_i^*, s) - d(v_i^*, \tilde v_i) < d(v_i^*, s)$; the rest have utility $-d(\tilde v_i, s)$. Only $k-1 < k$ entries equal $d(\tilde v_i, s)$, so the $k$-th largest is at most $d(v_i^*, s) - d(v_i^*, \tilde v_i) < d(v_i^*, s)$.
\end{itemize}
Hence $\phi_\sigma(u(v_i^*)) > \phi_\sigma(u(\tilde v_i))$ and $v_i^*$ wins. The epoch quiesces at $v_i^*$, and member~$i$'s true utility is $d(v_i^*, s)$.

\emph{Misreport case.} Member~$i$ votes $\tilde v_i$. Total votes: $k-1$ at $v_i^*$, $k$ at $\tilde v_i$, $n - 2k + 1$ at $s$.

If $\tilde v_i \neq s$: by symmetric analysis, $\phi_\sigma(u(\tilde v_i)) = d(\tilde v_i, s) > 0$ and $\tilde v_i$ wins. Member~$i$'s true utility for $\tilde v_i$ is $d(v_i^*, s) - d(v_i^*, \tilde v_i) < d(v_i^*, s)$.

If $\tilde v_i = s$: proposal $\tilde v_i = s$ has all-zero utilities, $\phi_\sigma = 0$, and is not supported; proposal $v_i^*$ has only $k-1$ positive entries (the voters at $v_i^*$, since member~$i$ now votes $s$), so it is not supported either. No proposal is supported, the status quo is retained, and member~$i$'s true utility is $0 < d(v_i^*, s)$.

In every case, the misreport yields strictly lower true utility than sincere voting. \qed
\end{proof}

\coronedprotocolsp*
\begin{proof}
By Proposition~\ref{proposition:1d-dominance}, the per-round winner is the proposal closest to the $\lceil \sigma n \rceil$-th positional vote, and public proposals do not alter it. Within an epoch votes are fixed, so the per-round winner is invariant across rounds; the epoch quiesces at this winner. Per-round dominance transfers to the epoch outcome. \qed
\end{proof}

\thmzerogap*
\begin{proof}
By Proposition~\ref{proposition:1d-characterisation}, for every $x \in X$ the median utility is $\phi_\sigma(u(x)) = |v_m - s| - |v_m - x|$, where $v_m$ is the positional median of $V$. This expression is uniquely maximised at $x = v_m \in V$, with value $|v_m - s|$. Hence $\text{OPT}_{\phi_\sigma}(V, s) = |v_m - s|$ is attained at a peak, and $\text{CG}_{\phi_\sigma}(V, s) = 0$. \qed
\end{proof}

\proplipschitz*
\begin{proof}
For every member with ideal $q$ and every pair of points $x^*, p \in X$,
\[
|u(q, x^*) - u(q, p)| = |d(q, p) - d(q, x^*)| \le d(x^*, p)
\]
by the reverse triangle inequality. Hence $\|u(x^*) - u(p)\|_\infty \le d(x^*, p)$. The $1$-Lipschitz hypothesis on $\phi$ gives $|\phi(u(x^*)) - \phi(u(p))| \le d(x^*, p)$. Minimising over $p \in V$ yields the claim. \qed
\end{proof}

\section{Related Work}
\label{appendix:related}

\mypara{Aggregation over metric spaces} Bulteau, Shahaf, Shapiro, and Talmon~\cite{bulteau2021aggregation} introduced the framework of aggregation over metric spaces, studying Condorcet aggregation and the $L_p$ family across six social choice settings. The present work uses their metric-space instantiations and admits the $L_p$ family as one source of public proposals (Section~\ref{section:gap}); the constitutional governance rule trades unconstrained optimality for polynomial-time existence and supermajority discipline.

\mypara{Reality-aware social choice} Shapiro and Talmon~\cite{shapiro2018incorporating} introduced the status quo as a distinguished, ever-present alternative in social choice, and showed how it breaks Condorcet cycles. This has been applied to sybil-resilient voting~\cite{shahaf2019sybil,meir2022sybil} and deliberative coalition formation~\cite{elkind2021united,elkind2022complexity}. The constitutional governance rule places the status quo at the centre: every proposal is compared only to it, yielding a scalar score that cannot cycle in any metric space (Theorem~\ref{theorem:condorcet}).

\mypara{Constitutional governance and consensus} Abramowitz, Shapiro, and Talmon~\cite{abramowitz2021beginning} studied the founding and amendment of constitutions, deriving from natural axioms that the initial decision rule must be simple majority and that the unique self-referential amendment process is the $h$-rule. Keidar, Lewis-Pye, Shapiro, and Talmon~\cite{keidar2025constitutional} instantiated this in the constitutional consensus protocol, combining democratic decision processes for amending the participant set, threshold, and timeout with a Byzantine-fault-tolerant consensus protocol. The present work incorporates the $\sigma$-parameterisation and the structural distinction between laws and the constitution, enabling unified treatment of decisions and amendments under the constitutional governance rule, with the $h$-rule reserved for the one self-referential case of $\sigma_C$ amending itself.

\mypara{Deliberative coalition formation and AI mediation} Elkind, Grossi, Shapiro, and Talmon~\cite{elkind2021united} model deliberation as the formation of coalitions around compromise proposals that can defeat the status quo; the complexity of deliberative coalition formation is studied in~\cite{elkind2022complexity}. AI mediation~\cite{briman2025ai} generates compromise proposals acceptable to disagreeing parties. The framework admits both as sources of public proposals, scored under the constitutionally-specified aggregation rule.

\mypara{Median spaces and strategy-proofness} Nehring and Puppe~\cite{nehring2007structure} characterised the domains on which strategy-proof, anonymous, and efficient social choice functions exist, identifying median spaces as the exact boundary; Brady and Chambers~\cite{brady2017spatial} provided axiomatic characterisations of the median as the unique rule satisfying anonymity and Maskin monotonicity. These results provide theoretical grounding for the generalised median instance beyond one dimension. Nehring and Puppe~\cite{nehring2019frugal} also proposed the frugal majority rule for budget allocation, philosophically close to our setting in starting from peaks alone; the constitutional governance rule differs by restricting the outcome to actual proposals, producing a unique winner, and anchoring on the status quo. Puppe and Rollmann~\cite{puppe2021mean} provided experimental support for median rules over mean rules in budget allocation.

\mypara{Metric distortion and facility location} The metric distortion literature~\cite{anshelevich2018approximating,goel2017metric,procaccia2006distortion} studies how well voting rules approximate optimal social choice when voters and alternatives lie in a common metric space; the median objective admits constant-factor approximation with ordinal information. Hakimi~\cite{hakimi1964optimum} established that on networks the facility location minimising total distance is always at a vertex---classical justification for restricting outcomes to discrete alternatives---and Procaccia and Tennenholtz~\cite{procaccia2009approximate} explored mechanism-design analogues. The compromise gap of Section~\ref{section:gap} is a metric-distortion-style quantity, measuring how much restriction to the proposal set costs.

\mypara{The spatial model} The spatial model of elections~\cite{SpatialModel} infers preferences from voter ideal points using an underlying metric, extended in many directions~\cite{peters1993generalized}. The framework adopts the spatial model's distance-induced preferences but differs in restricting solutions to member proposals and centring the status quo.

\section{Anonymity and Neutrality}
\label{appendix:anonymity}

The rule is \emph{neutral} among proposals by construction: it depends on a proposal $p$ only through its utility vector $u(p)$, which is computed from $d$, $s$, and $V$. No proposal-specific label, identity, or exogenous weight enters. The rule is \emph{anonymous} whenever $\phi$ is a symmetric function of its arguments: under this hypothesis, the rule depends on $V$ only through the multiset of utilities. The generalised median is symmetric, hence the median instance is anonymous.

\section{Finite Termination: Verification}
\label{appendix:termination-proof}

The proof of Proposition~\ref{proposition:termination} is in Appendix~\ref{appendix:proofs}.

The seven application settings of Section~\ref{section:metrics} are all totally bounded: bounded subsets of $\mathbb{R}$ are totally bounded; the simplex is compact; the discrete, permutation, and subset metrics are defined on finite sets; and balls in $\Sigma^*$ under weighted Levenshtein distance are finite, since strings within a bounded edit distance of a fixed string over a finite alphabet are finitely many.

\section{Monotonicity}
\label{appendix:monotonicity-counterexample}

\begin{definition}[Vote-to-outcome monotonicity]
\label{definition:monotonicity}
An aggregation method is \emph{monotone} if, for every vote profile $V$ and every co-winner $w$, changing a member's ideal element to $w$ does not cause $w$ to lose.
\end{definition}

\begin{restatable}{proposition}{proponedmonotonicity}
\label{proposition:1d-monotonicity}
In one-dimensional settings (with the absolute-value metric) and $\sigma = 1/2$, the constitutional governance rule with the generalised median is monotone.
\end{restatable}

\begin{proof}
The median utility of a proposal $p$ equals $u(v_m, p) = |v_m - s| - |v_m - p|$, where $v_m$ is the positional median of the votes (Proposition~\ref{proposition:1d-characterisation}). When member $i$ changes from $v_i$ to $w$ (a co-winner), $v_m$ moves towards $w$ or stays fixed; let $v_m'$ denote the new median. For $w$: $|v_m' - w| \le |v_m - w|$, so the median utility of $w$ does not decrease. For any $p \neq w$ that was a co-winner: $|v_m - p| \ge |v_m - w|$ before the move, and $|v_m' - p| \ge |v_m' - w|$ after (since $v_m'$ moves towards $w$, $p$ remains weakly farther from $v_m'$ than $w$ is). Hence the median utility of $w$ remains at least that of $p$ in the new profile. \qed
\end{proof}

In general metric spaces, monotonicity does not hold; the counterexample below uses seven members and exhibits a profile in which the winner $w$ loses after a member moves to $w$. The mechanism is that the departing member's utility for a competing proposal sits at the median position; when the member moves, the next value in the sorted order is revealed at the median and exceeds the unchanged median utility of $w$.

\mypara{Public-proposal-channel monotonicity} The framework separates two roles. The vote is an immutable, sealed expression of preference. The public proposal is a submission a member is willing to support. Vote-to-outcome monotonicity (Definition~\ref{definition:monotonicity}) is the right axiom for rules in which the vote is the only support mechanism (as in Bulteau et al.~\cite{bulteau2021aggregation}, whose $L_p$ aggregator is non-monotone for every $p > 1$). Under our framework, where support is expressed via the public-proposal channel, the relevant axiom is \emph{public-proposal-channel monotonicity}: submitting a proposal as a public proposal never decreases the submitter's utility. This holds for the constitutional governance rule by construction, regardless of metric space dimension---admitting a public proposal leaves all other proposals' scores unchanged, so either the winner is unchanged or it becomes the submitted public proposal. The vote-to-outcome failure in higher dimensions therefore identifies a real but axiom-mismatched concern; the framework's two-channel design answers it.

\subsection{Counterexample}

We show that the constitutional governance rule with $\sigma = 1/2$ is not monotone in general metric spaces. The counterexample uses seven members and a metric space with points $s, w, p, v_1, \ldots, v_5$.

\mypara{Setup} Let $d(w, s) = 10$, $d(p, s) = 10$, $d(w, p) = 5$, and $d(v_j, s) = 10$ for $j = 1, \ldots, 5$. The remaining distances are:
\begin{center}
\begin{tabular}{@{}cccc@{}}
\toprule
Member & $d(v_j, w)$ & $u(v_j, w)$ & $u(v_j, p)$ \\
\midrule
1 & 13 & $-3$ & $-6$ \\
2 & 12 & $-2$ & $-4$ \\
3 & 11 & $-1$ & $-2$ \\
4 & 9.99 & $0.01$ & $-0.01$ \\
5 & 9.98 & $0.02$ & $0.1$ \\
6 ($= p$) & 5 & $5$ & $10$ \\
7 ($= w$) & 0 & $10$ & $5$ \\
\bottomrule
\end{tabular}
\end{center}
All distances satisfy the triangle inequality; inter-member distances are realised via shortest paths through $\{s, w, p\}$.

\mypara{Before the move} Sorted utilities for $w$: $(-3, -2, -1, 0.01, 0.02, 5, 10)$; median~=~$0.01$. Sorted utilities for $p$: $(-6, -4, -2, -0.01, 0.1, 5, 10)$; median~=~$-0.01$. Every other proposal $v_j$ ($j = 1, \ldots, 5$) has only one supporter and median utility well below~$0$. Thus $w$ is the unique winner, with $\phi_\sigma(u(w)) = 0.01 > -0.01 = \phi_\sigma(u(p))$.

\mypara{Member 4 moves to $w$} Member~4 changes its ideal element from $v_4$ to $w$. The new utilities for member~4 are: $u(w, w) = 10$ and $u(w, p) = 5$.

Sorted utilities for $w$: $(-3, -2, -1, 0.02, 5, 10, 10)$; median~=~$0.02$. Member~4's value jumped from $0.01$ to $10$, but member~5's value $0.02$ now holds the median position. Sorted utilities for $p$: $(-6, -4, -2, 0.1, 5, 5, 10)$; median~=~$0.1$. Member~4's value jumped from $-0.01$ to $5$, and member~5's value $0.1$ now occupies the median position.

Now $\phi_\sigma(u(p)) = 0.1 > 0.02 = \phi_\sigma(u(w))$, so $w$ loses to~$p$. Monotonicity is violated.

\mypara{Mechanism} The triangle inequality guarantees that when a member moves to the winner $w$, its utility for $w$ increases at least as much as its utility for any competitor $p$. However, in general metric spaces the median can fail to respond: if member~4's utility for $w$ was already at a ``plateau'' of small positive values, its increase does not shift the median of $w$ correspondingly. Meanwhile, member~4's utility for $p$ was at the median position, and its departure reveals the next value in the sorted order, which exceeds the new median of~$w$. In one-dimensional settings, this cannot occur: the positional median moves monotonically towards $w$ when a member moves to $w$, preventing any competitor from overtaking.

\section{Mean Aggregation as an Alternative Instantiation}
\label{appendix:mean}

The framework is parametrised by an aggregation rule $\phi$ (Definition~\ref{definition:rule}); Section~\ref{section:rules} develops the generalised median. The mean is an alternative utilitarian instantiation. It satisfies all framework-level hypotheses of Section~\ref{section:framework_properties}---$O(n)$-computable via direct summation, symmetric, and zero at indifference---and inherits anonymity, reality-awareness, and Condorcet-cycle immunity (Theorem~\ref{theorem:condorcet}).

\begin{definition}[Mean]
\label{definition:mean}
The \emph{mean} aggregation rule is $\phi_{\text{mean}}(u) := \frac{1}{n} \sum_{i=1}^n u_i$.
\end{definition}

The mean is the natural utilitarian aggregator: it ranks proposals by total member gain over the status quo, equivalently by aggregate $L_1$-distance reduction. In the terminology of Bulteau et al.~\cite{bulteau2021aggregation}, the mean instance corresponds to $L_1$ aggregation under the supermajority gate.

Unlike the median, the mean does \emph{not} coincide with the supermajority gate: a proposal can be supported (i.e., strictly preferred to the status quo by at least $\lceil \sigma n \rceil$ members) yet have negative mean utility, or have positive mean yet fail the gate. The framework's gate disciplines the mean: a proposal strongly preferred by a minority but opposed by a supermajority is \emph{not} adopted under the mean instance, because the gate excludes it. This forecloses the classical tyranny-of-the-intense-minority objection to utilitarianism.

\mypara{Majoritarity} The mean instance is $\sigma$-majoritarian under the framework's gate: if $\lceil \sigma n \rceil$ members share $w \neq s$, then $w$ is supported, and---unless a distinct supported proposal has higher mean---$w$ is adopted. Without the gate (i.e., under pure utilitarian mean aggregation as in~\cite{bulteau2021aggregation}), the mean instance would not be $\sigma$-majoritarian: a minority strongly preferring $s$ could overwhelm a supermajority's aggregate gain. The gate forecloses this case.

\mypara{Per-round non-strategy-proofness} The mean instance is non-SP per round: a member can shift the aggregate by misreporting an exaggerated position, the standard utilitarian-voting weakness.

\mypara{Median or mean} A community choosing between the generalised median and the mean is weighing classical utilitarianism against median-voter representation. The mean maximises aggregate utility and is the natural choice when the community's objective is the sum of members' gains. The generalised median protects against tyranny-of-the-minority and is the natural choice when each decision must command supermajority support of broadly aligned values. Either is a constitutional choice, amendable via the framework's self-amendment mechanism (Section~\ref{subsection:self-amendment}).

\section{Per-Round Anatomy of Manipulation under $\phi_\sigma$}
\label{appendix:per-round-anatomy}

The following three lemmas characterise the structural anatomy of single-round manipulation under $\phi_\sigma$: every successful misreport falls into one of two cases, each of which has its own diagnostic.

\begin{restatable}[Self-defeating flip from unsupported]{lemma}{lemflipunsupported}
\label{lemma:flip-unsupported}
Fix $\sigma \in [1/2, 1)$ and a proposal $W$ that is not supported under sincere voting. Suppose member $i$ unilaterally misreports, and under the misreport $W$ becomes supported. Then $u(v_i^*, W) \le 0$: member~$i$ weakly prefers the status quo to $W$.
\end{restatable}

\begin{proof}
Let $k = \lceil \sigma n \rceil$. $W$ is supported iff at least $k$ entries of $u(W)$ are strictly positive. Under sincere voting, at most $k-1$ entries of $u(W)$ are positive. Misreporting changes only member~$i$'s entry, so the count of positive entries can increase by at most one---from $\le 0$ to $> 0$. A flip from unsupported to supported therefore requires $u(v_i^*, W) \le 0$. \qed
\end{proof}

\begin{restatable}[Public proposal replicates supported flip]{lemma}{lemcompromisereplicates}
\label{lemma:compromise-replicates}
Fix any proposal $c \neq s$ with $u(v_i^*, c) > 0$ and any votes of the other $n-1$ members. Then $c$ is supported under sincere voting by member~$i$ with public proposal $c_i = c$ if and only if $c$ is supported under the misreport $\tilde v_i = c$.
\end{restatable}

\begin{proof}
Support depends only on the multiset $\{u(v_j, c) : j \in [n]\}$. The entries for $j \neq i$ are identical in both scenarios. The entry from member~$i$ is $u(v_i^*, c) > 0$ in the public-proposal scenario by hypothesis, and $u(c, c) = d(c, s) > 0$ in the misreport scenario since $c \neq s$. Both entries are strictly positive, so the count of positive entries---and hence the support status of $c$---is the same. \qed
\end{proof}

\begin{restatable}[Self-defeating winner swap]{lemma}{lemwinnerswap}
\label{lemma:winner-swap}
Suppose under sincere voting, $W_0$ and $\tilde W$ are both supported with $\phi_\sigma(u(W_0)) > \phi_\sigma(u(\tilde W)) > 0$, so $W_0$ wins. Suppose member~$i$ unilaterally misreports, and under the misreport both $\phi_\sigma(u(W_0))$ strictly decreases and $\phi_\sigma(u(\tilde W))$ strictly increases. Then $u(v_i^*, W_0) > u(v_i^*, \tilde W)$: member~$i$ strictly prefers $W_0$ to $\tilde W$.
\end{restatable}

\begin{proof}
When a single entry in a list of $n$ reals is replaced, the $\lceil \sigma n \rceil$-th largest can strictly decrease only if the replaced entry was at least the original $\lceil \sigma n \rceil$-th largest, and can strictly increase only if it was at most. So a strict decrease of $\phi_\sigma(u(W_0))$ implies $u(v_i^*, W_0) \ge \phi_\sigma(u(W_0))$; a strict increase of $\phi_\sigma(u(\tilde W))$ implies $u(v_i^*, \tilde W) \le \phi_\sigma(u(\tilde W))$. Combining with $\phi_\sigma(u(W_0)) > \phi_\sigma(u(\tilde W))$:
\[ u(v_i^*, W_0) \ge \phi_\sigma(u(W_0)) > \phi_\sigma(u(\tilde W)) \ge u(v_i^*, \tilde W). \qed \]
\end{proof}

Together, Lemmas~\ref{lemma:flip-unsupported}--\ref{lemma:winner-swap} establish a per-round dichotomy: whenever a member could profitably misreport, either (a)~the manipulation is self-defeating---the member would weakly prefer the unmanipulated outcome (Lemmas~\ref{lemma:flip-unsupported},~\ref{lemma:winner-swap}), or (b)~the same outcome is achievable by sincere voting with a public-proposal submission (Lemma~\ref{lemma:compromise-replicates}). The per-round dichotomy motivates the design of the proposal protocol of Section~\ref{subsection:protocol-sp}: sealing votes denies the misreporter the information needed to construct case-(a) manipulations, and the open public-proposal channel makes case (b) sincere by construction. Theorem~\ref{theorem:protocol-sp} establishes the resulting protocol-level guarantee directly.

\begin{remark}
A concrete example illustrates that the public-proposal channel cannot always replicate a profitable misreport. Take $\mathbb{R}^2$ with Euclidean distance, $n = 3$, $s = (0,0)$, $\sigma = 1/2$, and ideals $v_1^* = (2.5, 1)$, $v_2^* = (1, 2)$, $v_3^* = (-1, -1)$. Under sincere voting, both $v_1^*$ and $v_2^*$ are supported; $v_2^*$ wins with median utility $\approx 0.890$ versus $v_1^*$'s $\approx 0.433$, so member~1's true utility is $\approx 0.890$. Under misreport $\tilde v_1 = (5, -5)$, $v_2^*$ becomes unsupported and $v_1^*$ wins, giving member~1 true utility $\approx 2.693$. By contrast, no public proposal $c$ submitted with sincere voting can lift $v_1^*$ above $v_2^*$: such a $c$ would need member~2's utility $u(v_2^*, c) > 0.890$, which constrains $c$ to a ball around $v_2^*$ excluding the part of $\mathbb{R}^2$ closest to $v_1^*$. The supremum of member~1's public-proposal-channel utility is $\approx 2.236 < 2.693$. A misreport thus achieves what public proposals cannot, by altering the manipulator's own entry in every proposal's median computation. The protocol of Section~\ref{subsection:protocol-sp}, by sealing votes at the start of the epoch, eliminates this possibility.
\end{remark}

\section{Public-Proposal and Coalition Behaviour}
\label{appendix:coalition}

The strategic role of the public-proposal channel is distinct from the voting channel. A member's vote is their immutable, sealed expression of preference; their public proposal carries supermajority public support under the revealed votes. The following propositions characterise rational behaviour on the public-proposal side.

\begin{restatable}[Public proposal is weakly beneficial]{proposition}{propcompromisebeneficial}
\label{proposition:compromise-beneficial}
Submitting a public proposal $c$ is weakly beneficial for the submitter: either the round winner is unchanged, or it becomes $c$. A rational member submits $c$ only when their true utility $u(v_i^*, c)$ exceeds their true utility for the previous round's winner.
\end{restatable}

\begin{proof}
Let $W$ be the round winner if member~$i$ submits no public proposal, and $W'$ the winner if they submit $c_i = c$. Since votes are fixed, admitting $c$ leaves $\phi$-scores of all other proposals unchanged, so $W' \in \{W, c\}$. If $W' = W$, the submitter's utility is unchanged; if $W' = c$, it becomes $u(v_i^*, c)$. A rational member therefore submits only when $u(v_i^*, c)$ exceeds their utility for $W$. \qed
\end{proof}

\begin{restatable}[Sincere coalition public proposal is weakly dominant]{proposition}{propcoalitioncompromise}
\label{proposition:coalition-compromise}
For any coalition $C \subseteq [n]$ and any joint choice of public-proposal submissions $\{c_j\}_{j \in C}$, each member's weakly dominant strategy is to submit a public proposal they personally prefer to the current winner. Submitting a public proposal they do not prefer is weakly dominated.
\end{restatable}

\begin{proof}
Each submission is evaluated independently: $\phi$ depends only on the fixed sincere votes, not on whose slot submitted the public proposal. The joint outcome is the proposal with maximum $u_\phi$ among all submissions. If member~$j$ submits $c_j$ with $u(v_j^*, c_j)$ not exceeding their utility for the current winner, the submission either has no effect (someone else's public proposal wins, or no public proposal wins) or causes $c_j$ to win, making $j$ weakly worse off. \qed
\end{proof}

\begin{restatable}[Coalition flip requires self-defeating member]{proposition}{propcoalitionflip}
\label{proposition:coalition-flip}
Let coalition $C \subseteq [n]$ jointly misreport, and suppose the misreport flips a proposal $W$ from unsupported under sincere voting to supported. Then at least one coalition member $m \in C$ has sincere utility $u(v_m^*, W) \le 0$: the coalition must include a member who weakly prefers the status quo to $W$.
\end{restatable}

\begin{proof}
Let $k = \lceil \sigma n \rceil$ and let $p_{\text{sincere}}$, $p_{\text{misreport}}$ denote the number of positive entries of $u(W)$ under each scenario. Non-coalition entries are unchanged, so $p_{\text{misreport}} - p_{\text{sincere}} = |M_C| - |S_C|$, where $S_C = \{j \in C : u(v_j^*, W) > 0\}$ and $M_C = \{j \in C : u(\tilde v_j, W) > 0\}$. A flip from $p_{\text{sincere}} \le k-1$ to $p_{\text{misreport}} \ge k$ requires $|M_C| - |S_C| \ge 1$, hence some $m \in M_C \setminus S_C$ has $u(v_m^*, W) \le 0$ and $u(\tilde v_m, W) > 0$. \qed
\end{proof}

A coalition flip of an unsupported proposal therefore must include a member acting against their sincere preference. Without side payments or external enforcement, this member is individually irrational, and the coalition is not self-enforcing.

\section{1D Characterisation of the Generalised Median}
\label{appendix:1d-characterisation}

\begin{restatable}[1D characterisation]{proposition}{proponedcharacterisation}
\label{proposition:1d-characterisation}
Let $k = \lceil \sigma n \rceil$, and order the votes as $v_{(1)} \le \cdots \le v_{(n)}$. The generalised median utility of a proposal $p \in X$ is
\[
\phi_\sigma(u(p)) = \begin{cases}
u(v_{(n-k+1)}, p) & \text{if } p > s, \\
u(v_{(k)}, p) & \text{if } p < s, \\
0 & \text{if } p = s.
\end{cases}
\]
Among supported proposals in $V$, the winner is the one closest to the corresponding positional voter. For $\sigma = 1/2$ with odd $n$, the two positional voters coincide at the median $v_m$, and the winner is the proposal in $V$ closest to $v_m$ (which equals $v_m$ itself when $v_m \in V$ and $v_m \neq s$).
\end{restatable}

\begin{proof}
For fixed $p$ and $s$, the utility $u(v_q, p) = |v_q - s| - |v_q - p|$ is non-decreasing in $v_q$ when $p > s$ and non-increasing when $p < s$. Hence the ordering of utilities by member index tracks (respectively reverses) the ordering of ideal elements by magnitude, and the $k$-th largest utility equals $u(v_{(n-k+1)}, p)$ when $p > s$ and $u(v_{(k)}, p)$ when $p < s$. In either direction, the resulting expression is of the form $|v_m - s| - |v_m - p|$ for the appropriate positional voter, maximised over $V$ at the proposal closest to $v_m$; if $v_m \in V$ and $v_m \neq s$, this proposal is $v_m$ itself, with $\phi_\sigma = |v_m - s| > 0$. \qed
\end{proof}

\begin{example}[1D dominance in the running cooperative]
\label{example:freelancer-1d}
A constitutional amendment in the running example sets the commission rate at $\sigma = 1/2$. The current rate is $s = 20\%$. Members vote $v_1 = 10\%, v_2 = 15\%, v_3 = 18\%, v_4 = 22\%, v_5 = 25\%$; the positional median is $v_3 = 18\%$. By Proposition~\ref{proposition:1d-dominance}, no member can profit from misreporting. For instance, member~1 ($v_1^* = 10\%$) shifting their report leftward to $5\%$ leaves the median at $18\%$; shifting rightward to $19\%$ shifts the median to $v_4 = 22\%$, farther from the true ideal. Sincere voting weakly dominates.
\end{example}

\section{Heuristic P: Definition and Complexity}
\label{appendix:heuristic-details}

\begin{definition}[Heuristic P]
\label{definition:heuristic-p}
Given proposal set $P^r$ at the start of a round, status quo $s$, and aggregation rule $\phi$:
\begin{enumerate}
    \item For each pair $(p, q) \in P^r \times P^r$ with $p \neq q$, compute a bounded set $C(p, q) \subseteq X$ of candidate combinations of $p$ and $q$ (in metric spaces with unique geodesic midpoints, $C(p, q)$ is the singleton midpoint; otherwise $C(p, q)$ enumerates or samples a bounded number of tie-breaking candidates).
    \item Let $c^* \in \arg\max\{\phi(u(c)) : c \in \bigcup_{p, q} C(p, q)\}$.
    \item If $\phi(u(c^*)) > \max_{p \in P^r} \phi(u(p))$, return $c^*$ as a public proposal; otherwise return $\bot$.
\end{enumerate}
\end{definition}

\begin{proposition}[Complexity of Heuristic P]
\label{proposition:heuristic-complexity}
Heuristic P runs in $O(|P^r|^2 (T_c + n))$ time, where $T_c$ bounds the cost of computing the candidate combinations for one pair: $O(m)$ for the simplex under Euclidean distance, $O(m^2)$ for permutations under swap distance, $O(|A|)$ for subsets under symmetric difference (bounded candidate set), and $O(\ell^2)$ for strings under weighted Levenshtein. Heuristic P is polynomial in every setting of Section~\ref{section:metrics}.
\end{proposition}

\section{Compromise Gap: Minimal Examples}
\label{appendix:gap-examples}

Two minimal examples illustrate that in dimensions two and above, the gap can be positive.

\begin{example}[Star graph]
\label{example:star}
Let $X$ be a star with three leaves $\ell_1, \ell_2, \ell_3$ each at distance $1$ from a hub $h$, with the shortest-path metric. Take $V = \{\ell_1, \ell_2, \ell_3\}$, $s = \ell_1$, $\sigma = 1/2$, so $k = \lceil 3/2 \rceil = 2$. The utility vector at $h$ is $u(h) = (-1, 1, 1)$ with $\phi_\sigma(u(h)) = 1$. For each peak $p \in V$, direct computation gives $\phi_\sigma(u(p)) = 0$ (with $u(\ell_1) = (0, 0, 0)$, $u(\ell_2) = (-2, 2, 0)$, $u(\ell_3) = (-2, 0, 2)$). Hence $\text{OPT} = 1$ at $h \notin V$, $\text{PEAK} = 0$, and $\text{CG} = 1$.
\end{example}

\begin{example}[Two-dimensional Euclidean]
\label{example:2d}
Proposition~\ref{proposition:multidim-not-sp} exhibits $n = 3$ peaks in $\mathbb{R}^2$ for which no peak is supported and the status quo is retained ($\text{PEAK} = 0$ in our terminology, after enforcing positivity), yet the non-peak point $(0.5, 0.5)$ has positive median utility and is supported. Hence $\text{CG} > 0$.
\end{example}

\begin{example}[Compromise gap in the running cooperative]
\label{example:freelancer-gap}
The budget profile of Example~\ref{subsection:running-example} has $\text{CG} > 0$. The five peaks score below the geometric-median public proposal $\hat c = (0.34, 0.40, 0.27)$: direct computation gives $\phi_{1/2}(u(\hat c)) \approx 0.171$, while no peak exceeds $\approx 0.140$ on the same profile. The public-proposal channel closed the gap; without it, the rule would have adopted the highest-scoring peak with score $\approx 0.140$, leaving $\text{CG} \approx 0.031$ unclaimed. This is the gap that the open public-proposal channel exists to close.
\end{example}

In both examples, an interior point of $X$ strictly dominates every peak.

\section{Remaining Application Settings}
\label{appendix:settings}

This appendix gives full propositions, proofs, and worked examples for four of the five application settings summarised in Section~\ref{subsection:other-settings}: plurality elections, social welfare functions, committee elections, and legislation. (1D rate-setting and simplex budgeting are worked in the main text; constitutional amendments are detailed in Appendix~\ref{appendix:amendment}.)

\subsection{Plurality Elections: Electing Officers}
\label{appendix:plurality}

A community must periodically elect officers---a chairperson, treasurer, or ombudsperson. Each member nominates a candidate; at the end of a term, the position is vacant.

\mypara{Model} Let $A$ be the set of candidates and $\bot$ denote the vacant position. The metric space is $X = A \cup \{\bot\}$ with the discrete metric: $d(x, y) = 1$ for $x \neq y$, and $d(x, x) = 0$. Each member's ideal element is one of the candidates. The status quo $s = \bot$ is the vacant position.

\mypara{Utility} Since $s = \bot$ and no member's ideal element is $\bot$, every member has $d(v_q, s) = 1$. For a member with ideal element $v_q$ and a candidate $p \in A$:
\[ u(q, p) = d(v_q, s) - d(v_q, p) = \begin{cases} 1 & \text{if } v_q = p, \\ 0 & \text{otherwise.} \end{cases} \]

\mypara{Characterisation} For a candidate $p \in A$, let $n_p = |\{q : v_q = p\}|$ be the number of supporters of $p$.

\begin{restatable}{proposition}{propplurality}
\label{proposition:plurality}
The constitutional governance rule selects the candidate with the most supporters, provided that candidate has $\sigma$-supermajority support, i.e., at least $\lceil \sigma n \rceil$ supporters. Otherwise no candidate is adopted and the position remains vacant.
\end{restatable}

\begin{proof}
For a candidate $p$, the $n$ utilities consist of $n_p$ values of $1$ and $n - n_p$ values of $0$. The $\lceil \sigma n \rceil$-th largest utility equals $1$ iff $n_p \ge \lceil \sigma n \rceil$, and $0$ otherwise. Among all candidates, the one with the most supporters maximises the median; it is adopted precisely when it meets the gate. \qed
\end{proof}

\begin{example}
A cooperative of five members elects a chairperson with $\sigma = 1/2$. The candidates are Alice, Bob, and Carol. If two members prefer Bob, two prefer Alice, and one prefers Carol, no candidate has $\lceil 5/2 \rceil = 3$ supporters; the position remains vacant. If three prefer Bob and the rest are split, Bob is elected.
\end{example}

If no candidate achieves the threshold, the public-proposal channel admits compromise candidates---e.g., a coalition-suggested candidate---which are scored under the same rule.

\subsection{Social Welfare Functions: Ranking Priorities}
\label{appendix:swf}

A community must rank its strategic priorities---growth, member welfare, service quality, sustainability---to guide resource allocation. Each member proposes a ranking; the current ranking is the status quo.

\mypara{Model} Let $A$ be the set of priorities, $m = |A|$. The metric space is $X = \mathbb{S}_A$ (permutations over $A$) with the swap distance (minimum number of adjacent transpositions to convert one permutation into another). The status quo $s \in \mathbb{S}_A$ is the current ranking.

\mypara{Utility} For member with ideal element $v_q$ and a proposed ranking $p$:
\[ u(q, p) = d(v_q, s) - d(v_q, p) . \]

\mypara{Characterisation}

\begin{restatable}{proposition}{propswf}
\label{proposition:swf}
The constitutional governance rule selects the supported ranking in $V \cup C$ with maximal score in $O(n^2 m \log m)$ time.
\end{restatable}

\begin{proof}
For each proposal $p \in P^r$, the swap distance $d(v_q, p)$ equals the number of inversions between the two permutations and is computable in $O(m \log m)$ time. Computing utilities for all $n$ members takes $O(n m \log m)$, evaluating $\phi_\sigma$ takes $O(n)$, so per-proposal cost is $O(n m \log m)$. Over $|P^r| \le 2n$ proposals, the total is $O(n^2 m \log m)$. \qed
\end{proof}

\mypara{Comparison with Bulteau et al.} Bulteau et al.'s $L_p$ aggregation here corresponds to the Kemeny ranking, which is NP-hard for all $p \ge 1$. The constitutional governance rule trades optimality over the full permutation space for tractability over the proposal set; the Kemeny ranking, computed as a public proposal, is admitted by the framework on the same footing as any other public proposal.

\begin{example}
Three members rank priorities $A = \{g, w, q\}$ (growth, welfare, quality) at $\sigma = 1/2$. The current ranking is $s : g \succ w \succ q$. The proposed rankings are $v_1 = v_2 = s$ and $v_3 = w \succ g \succ q$ (one swap from $s$). The proposal $v_3$ has utility vector $(-1, -1, 1)$, with median $-1$. The current ranking has all-zero utilities. No proposal is supported with positive median; the status quo is retained, reflecting that two members already prefer it.
\end{example}

\subsection{Committee Elections: Selecting a Board}
\label{appendix:committee}

A community must select a governance board, a dispute panel, or an audit committee from among its members.

\mypara{Model} Let $A$ be the set of eligible members. The metric space is $X \subseteq 2^A$ (subsets of $A$) with $d(x, y) = |x \triangle y|$, the symmetric difference. As specific cases: $X = 2^A$ for unrestricted board size, or $X = \binom{A}{k}$ for a board of exactly $k$ members. The status quo $s = \emptyset$ is the empty set (no board); the constitution may specify continuity mechanisms between terms.

\mypara{Utility}
\[ u(q, p) = |v_q \triangle s| - |v_q \triangle p| . \]

\mypara{Characterisation}

\begin{restatable}{proposition}{propcommittee}
\label{proposition:committee}
The constitutional governance rule selects the supported board with maximal score in $O(n^2 |A|)$ time.
\end{restatable}

\begin{proof}
For each proposal $p \in P^r$, the symmetric difference $|v_q \triangle p|$ is computable in $O(|A|)$ time. Per-proposal cost is therefore $O(n |A|)$, and over $|P^r| \le 2n$ proposals the total is $O(n^2 |A|)$. \qed
\end{proof}

\mypara{Comparison with Bulteau et al.} Bulteau et al.'s Condorcet aggregation is coNP-hard, and $L_p$ for $p > 1$ is NP-hard, in this setting. The framework's restriction to $P^r$ is again the source of tractability.

\begin{example}
Five members ($A = \{a, b, c, d, e\}$) form a three-member governance board at $\sigma = 1/2$, with $s = \emptyset$. Proposed boards: $v_1 = v_2 = \{a, b, c\}$, $v_3 = \{a, c, d\}$, $v_4 = \{b, c, d\}$, $v_5 = \{a, b, d\}$. The framework evaluates each proposal by its generalised median utility relative to the status quo of no board, and selects the supported board with the highest score.
\end{example}

\subsection{Legislation: Drafting Bylaws}
\label{appendix:legislation}

A community must draft and amend its bylaws---the rules governing membership, decision-making, dispute resolution, and surplus distribution. Each member may propose a text; the current bylaws (or the empty text, if none exist) are the status quo. The framework's metric-agnosticism (Remark~\ref{remark:semantic-metrics}) is most consequential here.

\mypara{Model} Let $\Sigma$ be an alphabet, where characters represent clauses or articles. The metric space is $X = \Sigma^*$ (strings over $\Sigma$). Following Bulteau et al.~\cite{bulteau2021aggregation}, we adopt a weighted Levenshtein distance: insert and delete operations cost $1$, swap operations cost $1/\ell^2$ where $\ell$ is the length of the longest proposed text. This weighting makes content more important than order. The status quo $s$ is the current bylaws (or the empty text).

As observed in Bulteau et al., the problem decomposes into two phases: Phase 1 selects the set of clauses (formally equivalent to committee elections, Appendix~\ref{appendix:committee}); Phase 2 orders the selected clauses (formally equivalent to social welfare functions, Appendix~\ref{appendix:swf}).

\mypara{Utility}
\[ u(q, p) = d(v_q, s) - d(v_q, p) . \]

\mypara{Characterisation}

\begin{restatable}{proposition}{proplegislation}
\label{proposition:legislation}
The constitutional governance rule selects the supported text with maximal score in $O(n^2 \ell^2)$ time.
\end{restatable}

\begin{proof}
For each proposal $p \in P^r$, the weighted Levenshtein distance $d(v_q, p)$ between two strings of length at most $\ell$ is computable in $O(\ell^2)$ time by dynamic programming. Per-proposal cost is therefore $O(n \ell^2)$, and over $|P^r| \le 2n$ proposals the total is $O(n^2 \ell^2)$. \qed
\end{proof}

\mypara{Comparison with Bulteau et al.} Both phases of $L_p$ aggregation are NP-hard, and the combined problem inherits the hardness. The framework's tractability rests on restricting to the proposal set.

\mypara{On the metric} As noted in Remark~\ref{remark:semantic-metrics}, the weighted Levenshtein distance is a placeholder. The framework's metric-agnosticism means semantic distances---defined via large language models, embedding-based similarity, or domain-specific schemas---plug in directly. AI mediation~\cite{briman2025ai} provides one path: a mediator computes a semantic public proposal between proposed bylaws and submits it under the open public-proposal channel. The rule scores the result under the constitutionally-specified metric, regardless of how the public proposal was generated.

\section{Constitutional Amendment: Full Details}
\label{appendix:amendment}

Section~\ref{subsection:self-amendment} sketched constitutional amendment as the same constitutional governance rule applied at the constitutional threshold $\sigma_C$, with the $h$-rule reserved for $\sigma_C$ amending itself. This appendix gives the per-component details. We organise around the components of the legal corpus: laws, scalar metadata, the supermajority threshold, the aggregation rule, and textual provisions.

\subsection{Components of the Legal Corpus}
\label{appendix:identify}

The legal corpus is a tuple of components, each carrying its own metric space $(X_i, d_i)$, aggregation rule $\phi_i$, supermajority threshold $\sigma_i$, and novelty distance $\varepsilon_i$. The components partition into the \emph{laws}---substantive components such as a current rate, a current budget, a current board, the current bylaws---and the \emph{constitution}---comprising the per-component metadata $(X_i, d_i, \phi_i, \sigma_i, \varepsilon_i)$ for every component, the member set, and the constitutional threshold $\sigma_C$. Each component is amended by a $\sigma_i$-supermajority via Definition~\ref{definition:rule}. The metadata of any law $i$, including $\sigma_i$, is amended by a $\sigma_C$-supermajority via Definition~\ref{definition:rule}. The constitutional threshold $\sigma_C$ amending itself is the one self-referential case, handled by the $h$-rule (Appendix~\ref{appendix:threshold}). New members must consent to joining; membership is amended via per-candidate binary referenda~\cite{keidar2025constitutional}, the simplest instance of Definition~\ref{definition:rule} on $\{\text{in},\text{out}\}$.

\subsection{Amending a Scalar Parameter}
\label{appendix:scalar}

Scalar constitutional metadata---such as the per-component novelty distance $\varepsilon_i$---are amended as instances of 1D elections (Section~\ref{subsection:rate}), at the constitutional threshold $\sigma_C$, with the metric space a subset of $\mathbb{R}$ with absolute distance. Each member proposes a value; the rule selects the proposal closest to the $\lceil \sigma_C n \rceil$-th positional vote, by Proposition~\ref{proposition:1d-characterisation}.

\subsection{Amending the Threshold: The h-Rule}
\label{appendix:threshold}

The most distinctive constitutional amendment is changing $\sigma$ itself. This is inherently self-referential: the rule used to decide whether $\sigma$ should change is determined by $\sigma$. Abramowitz et al.~\cite{abramowitz2021beginning} show that a natural set of axioms---decisiveness, monotonicity, anonymity, concordance, minimality, and posterior consistency---uniquely determines the amendment process, the \emph{$h$-rule}.

The $h$-rule operates as follows. Each member $i$ states a preferred threshold $\sigma_i \in [1/2, 1)$. Preferences over thresholds are single-peaked.

\begin{enumerate}
    \item \emph{Raising $\sigma$}: $\sigma$ is increased to the maximal $\sigma' > \sigma$ for which a $\sigma'$-supermajority of members voted for a value $\ge \sigma'$. Raising the threshold requires the very supermajority that the new threshold demands.
    \item \emph{Lowering $\sigma$}: $\sigma$ is decreased to the minimal $\sigma' < \sigma$ for which a $\sigma$-supermajority voted for a value $\le \sigma'$. Lowering requires the current supermajority's agreement.
    \item Otherwise $\sigma$ remains unchanged.
\end{enumerate}

The $h$-rule prevents a minority from imposing stricter requirements or weakening the constitution; \cite{abramowitz2021beginning} establishes that it is the unique rule satisfying their axioms.

\mypara{Running example: raising the constitutional threshold} Suppose a cooperative's current constitutional threshold is $\sigma_C = 1/2$, and members propose to raise it. Members vote preferred values: $v_1 = 1/2, v_2 = 2/3, v_3 = 2/3, v_4 = 3/4, v_5 = 3/4$. Raising to $\sigma' = 3/4$ requires a $3/4$-supermajority voting for $\ge 3/4$: $\lceil (3/4) \cdot 5 \rceil = 4$ members would be needed, but only $2$ vote $\ge 3/4$. Raising to $\sigma' = 2/3$ requires a $2/3$-supermajority voting for $\ge 2/3$: $\lceil (2/3) \cdot 5 \rceil = 4$ members are needed, and $4$ ($v_2, v_3, v_4, v_5$) qualify. So $\sigma_C$ is raised to $2/3$. The example illustrates the $h$-rule's defining property: a stricter threshold cannot be imposed without the level of support it itself demands.

This setting is an instance of 1D elections on $[1/2, 1)$, but with the crucial difference that the threshold for amendment is itself the value being amended; the $h$-rule resolves the self-reference.

\subsection{Amending the Aggregation Rule}
\label{appendix:aggregation}

Each amendable component $s_i$ carries an aggregation rule $\phi_i$, part of the constitution. A community may constitutionally specify $\phi_i$ as the generalised median (Definition~\ref{definition:median}), the mean (Definition~\ref{definition:mean}), another aggregator from Bulteau et al.'s $L_p$ family, or any anonymous, monotone, $O(n)$-computable aggregator with $\phi(\mathbf{0}) = 0$. Amending $\phi_i$ proceeds under the constitutional governance rule at $\sigma_C$. Each member proposes a preferred $\phi_i$; the rule selects the round winner. A community can evolve any component's aggregation rule from generalised median to mean (or to any admissible aggregator) as collective judgement evolves.

\subsection{Amending Textual Provisions}
\label{appendix:textual-amendment}

The community's bylaws are a law. Amending them is an instance of legislation (Appendix~\ref{appendix:legislation}) at the bylaws' own $\sigma_i$. Each member proposes a revised text; the constitutional governance rule selects the proposal with maximal score, adopting it only if it is supported.

\subsection{Ratification}
\label{appendix:ratification}

Once the constitutional governance rule produces an amendment, the constitutional consensus protocol~\cite{keidar2025constitutional} ratifies it. Constitutional consensus operates in epochs, each governed by the prevailing constitution. A constitutionally-valid amendment is submitted as a transaction, ratified by the consensus protocol, and takes effect in the next epoch. The protocol integrates the decision processes described above with a Byzantine-fault-tolerant consensus mechanism, ensuring amendments are both democratically decided and reliably enacted. We treat constitutional consensus as a black box for the purposes of this paper.

\section{Simulation: Full Sweep}
\label{appendix:simulation-full}

Section~\ref{subsection:simulation} reported the headline statistics, table, and findings. This appendix gives the complete sweep across all twenty-seven configurations: four settings (2D Euclidean, simplex, hypercube, permutations), each at three or four electorate sizes.

\begin{table}[h]
\caption{Compromise gap and Heuristic P across four settings and three electorate sizes. ``CG~$> 0$'' is the fraction of profiles (with $\text{OPT} > 0$) for which $\text{CG} > 0$. ``Gap-closing'' is the mean fraction of $\text{CG}$ closed by Heuristic P on profiles with $\text{CG} > 0$ on which P returns a non-$\bot$ output. ``Hit rate'' is the fraction of profiles on which P returns a non-$\bot$ output. ``Combined'' is their product, the effective gap closure over all profiles with $\text{CG} > 0$.}
\label{table:simulation-full}
\centering
\small
\begin{tabular}{@{}llrrrr@{}}
\toprule
Setting & $n$ & CG~$> 0$ & Gap-closing & Hit rate & Combined \\
\midrule
2D Euclidean (uniform) & 5 & 98.6\% & 41.1\% & 66.8\% & 27.5\% \\
2D Euclidean (uniform) & 11 & 99.6\% & 61.1\% & 87.8\% & 53.7\% \\
2D Euclidean (uniform) & 21 & 100\% & 79.4\% & 96.3\% & 76.5\% \\
2D Euclidean (uniform) & 51 & 96.0\% & 96.3\% & 98.0\% & 94.4\% \\
\midrule
Simplex ($m = 3$) & 5 & 93.0\% & 40.1\% & 60.0\% & 24.1\% \\
Simplex ($m = 3$) & 11 & 98.0\% & 64.4\% & 85.7\% & 55.2\% \\
Simplex ($m = 3$) & 21 & 98.0\% & 81.6\% & 95.0\% & 77.5\% \\
Simplex ($m = 4$) & 5 & 97.0\% & 52.4\% & 77.5\% & 40.6\% \\
Simplex ($m = 4$) & 11 & 99.0\% & 66.0\% & 94.0\% & 62.0\% \\
Simplex ($m = 4$) & 21 & 99.0\% & 81.1\% & 96.0\% & 77.9\% \\
\midrule
Hypercube ($|A| = 6$) & 5 & 39.4\% & 99.5\% & 39.2\% & 39.0\% \\
Hypercube ($|A| = 6$) & 11 & 26.2\% & 99.2\% & 26.0\% & 25.8\% \\
Hypercube ($|A| = 6$) & 21 & 11.7\% & 100\% & 11.7\% & 11.7\% \\
Hypercube ($|A| = 8$) & 5 & 68.5\% & 96.5\% & 66.3\% & 64.0\% \\
Hypercube ($|A| = 8$) & 11 & 51.4\% & 95.9\% & 49.4\% & 47.4\% \\
Hypercube ($|A| = 8$) & 21 & 35.5\% & 100\% & 35.5\% & 35.5\% \\
Hypercube ($|A| = 10$) & 5 & 85.0\% & 90.2\% & 78.7\% & 71.0\% \\
Hypercube ($|A| = 10$) & 11 & 73.0\% & 91.4\% & 67.5\% & 61.7\% \\
Hypercube ($|A| = 10$) & 21 & 51.0\% & 96.1\% & 49.0\% & 47.1\% \\
\midrule
Permutations ($m = 4$) & 5 & 27.4\% & 98.8\% & 27.1\% & 26.8\% \\
Permutations ($m = 4$) & 11 & 13.0\% & 97.3\% & 12.7\% & 12.4\% \\
Permutations ($m = 4$) & 21 & 4.8\% & 100\% & 4.8\% & 4.8\% \\
Permutations ($m = 5$) & 5 & 63.0\% & 81.8\% & 52.9\% & 43.3\% \\
Permutations ($m = 5$) & 11 & 47.5\% & 89.1\% & 42.5\% & 37.9\% \\
Permutations ($m = 5$) & 21 & 31.7\% & 98.9\% & 31.3\% & 31.0\% \\
Permutations ($m = 6$) & 5 & 85.0\% & 67.1\% & 65.3\% & 43.8\% \\
Permutations ($m = 6$) & 11 & 83.3\% & 77.4\% & 69.7\% & 53.9\% \\
Permutations ($m = 6$) & 21 & 72.0\% & 84.3\% & 62.7\% & 52.9\% \\
\bottomrule
\end{tabular}
\end{table}

\mypara{Reading the sweep} The pattern from the headline table generalises across all configurations. In continuous spaces (2D Euclidean, simplex), $\text{CG} > 0$ is near-universal at small $n$ and shrinks slowly with $n$; Heuristic P's gap-closing ratio rises from $\sim 40$--$50\%$ at $n = 5$ to $\sim 80$--$96\%$ at $n = 21$--$51$. In discrete spaces (hypercube, permutations), $\text{CG} > 0$ is less common---peaks themselves often coincide with the unconstrained optimum---and Heuristic P closes nearly the entire gap when there is one, with the gap-closing ratio frequently at or above $90\%$ across all configurations. The hypercube, a median space, exhibits the cleanest behaviour: gap-closing $\ge 90\%$ in every row.

\mypara{Scope} The simulation does not claim a worst-case approximation guarantee for Heuristic P; configurations are drawn from natural uniform-on-domain priors. Non-uniform priors, larger $m$ or $|A|$, alternative aggregators (e.g., the mean instance), and empirical evaluation on real cooperative-governance data are left as future work.

\end{document}